%% file: main.tex
\documentclass[twocolumn]{autart}

\usepackage[T1]{fontenc}
\usepackage{graphicx}
\usepackage{subfiles}
\usepackage{amsmath}
\usepackage{amsfonts}
\usepackage{amssymb}
\usepackage{mathtools}
\usepackage{physics}
\usepackage{url}
\usepackage{booktabs}
\usepackage{multirow}
\usepackage{tablefootnote}
\usepackage{xcolor}
\usepackage{tabularx}

\newtheorem{cond}[thm]{Condition}

\newcommand{\placeqed}{\nobreak\hfill \ensuremath{\blacksquare}}

\begin{document}

\begin{frontmatter}

\title{Data-driven augmentation of first-principles models under constraint-free well-posedness and stability guarantees\thanksref{footnoteinfo}}

\thanks[footnoteinfo]{This paper was not presented at any IFAC meeting. Corresponding author B.~Gy\"or\"ok. This project has been supported by the Air Force Office of Scientific Research under award number FA8655-23-1-7061, by The MathWorks Inc., and by the European Union (Horizon Europe, ERC, COMPLETE, 101075836). Views and opinions expressed are however those of the author(s) only and do not necessarily reflect those of the European Union or the European Research Council Executive Agency. Neither the European Union nor the other granting authorities can be held responsible for them.}

\author[SZTAKI]{Bendeg\'uz Gy\"or\"ok}\ead{gyorokbende@sztaki.hu}, %
\author[TUe]{Roel Drenth}\ead{r.drenth@tue.nl}, %
\author[TUe,UPenn]{Chris Verhoek}\ead{c.verhoek@tue.nl}, %
\author[SZTAKI]{Tam\'as P\'eni}\ead{peni@sztaki.hu}, %
\author[TUe]{Maarten Schoukens}\ead{m.schoukens@tue.nl}, %
\author[SZTAKI,TUe]{Roland T\'oth}\ead{r.toth@tue.nl}

\address[SZTAKI]{Systems and Control Laboratory, HUN-REN Institute for Computer Science and Control, Budapest, Hungary}%
\address[TUe]{Control Systems Group, Eindhoven University of Technology, Eindhoven, The Netherlands}%
\address[UPenn]{Department of Electrical and Systems Engineering, University of Pennsylvania, Philadelphia, United States of America}%


\begin{abstract}
The integration of first-principles models with learning-based components, i.e., model augmentation, has gained increasing attention, as it offers higher model accuracy and faster convergence properties compared to black-box approaches, while generating physically interpretable models. Recently, a unified formulation has been proposed that generalizes existing model augmentation structures, utilizing linear fractional representations (LFRs). However, several potential benefits of the approach remain underexplored. In this work, we address three key limitations. First, the added flexibility of LFRs also introduces possible algebraic loops, i.e., a problem of well-posedness. To address this challenge, we propose a constraint-free direct parametrization of the model structure with a well-posedness guarantee. Second, we introduce a constraint-free parametrization that ensures stability of the overall model augmentation structure via contraction. Third, we adopt an efficient identification pipeline capable of handling non-smooth cost functions, such as group-lasso regularization, which facilitates automatic model order selection and discovery of the required augmentation configuration. These contributions are demonstrated on various simulation and benchmark identification examples.
\end{abstract}

\end{frontmatter}

\subfile{sections/1_introduction.tex}
\subfile{sections/2_LFR_augmentation.tex}
\subfile{sections/3_WP_parametrization.tex}
\subfile{sections/4_stability.tex}
\subfile{sections/5_ident_algorithm.tex}
\subfile{sections/6_num_examples.tex}
\subfile{sections/7_conclusion.tex}
\subfile{sections/appendix.tex}
\bibliographystyle{abbrv}        
\bibliography{literature}           

\end{document}

%% file: sections/1_introduction.tex
\section{Introduction}\label{sec:intro}\vspace{-6pt}
Combining first-principles models with learning-based components has emerged as a promising research direction, largely due to the inherent limitations of purely data-driven approaches. One of the most important drawbacks of black-box models is that prior physical insight of the system is disregarded, often requiring significant effort to rediscover relations that are already well understood due to physics-based knowledge. 
Although deep-learning-based identification techniques have shown the ability to achieve extraordinary model accuracy~\cite{beintema_deep_2023} and high generalization capabilities~\cite{forgione_system_2023}, their opaque nature typically limits their applicability in contexts where interpretability and physical insight are essential, e.g., in the control design process and safety-critical applications~\cite{ljung_perspectives_2010}.

A traditional method for addressing this challenge is (light) gray-box modeling~\cite{bohlin_practical_2006}. Light gray-box models utilize a fixed physics-based structure and apply data-driven parameter estimation, often relying on statistical learning algorithms~\cite{aghmasheh_gray_2017}, but deep-learning-based techniques are also used~\cite{takeishi_deep_2023}. A more recent approach for incorporating \emph{first-principles} (FP) knowledge into data-driven identification is based on \emph{physics-informed neural networks} (PINNs) and \emph{physics-guided neural networks}~(PGNNs), introduced in~\cite{raissi_physics-informed_2019} and~\cite{daw_physics-guided_2022}. In both branches of modeling approaches, physical knowledge is typically incorporated through additional terms in the cost function, penalizing when neural networks violate prescribed physical laws during training~\cite{karpatne_theory-guided_2017}.

Gray-box models and PINNs can be viewed as two extremes of the same spectrum. Light gray-box models rely on strictly physics-based structures with unknown parameters to be estimated; on the other hand, PINNs offer a flexible black-box representation that is constrained only by selected physical laws. An alternative, promising direction of utilizing physical knowledge in nonlinear system identification is \emph{model augmentation}, which integrates FP models and learning components in a structural manner. Such augmented models not only result in physically interpretable structures, but faster convergence and improved model accuracy can be typically achieved compared to purely black-box learning methods~\cite{hoekstra_learning-based_ECC_2025}. Furthermore, existing approaches can be incorporated into the model augmentation framework, e.g., physics-based loss functions can also be applied~\cite{daw_physics-guided_2022}, and physical parameters can be co-estimated with the parameters of the learning component~\cite{psichogios_hybrid_1992}. 

However, the combination of the FP part and the data-driven part is not intuitive, as many different model augmentation structures are available in the literature, e.g., see~\cite{sun_comprehensive_2020,gotte_composed_2022}. To unify these methods, a general formulation has been proposed based on \emph{Linear Fractional Representations}~(LFRs) in~\cite{hoekstra_learning-based_ECC_2025} with binary selection matrices connecting the baseline and learning components. The (binary) values of the interconnection weights can be chosen, i.e., fixed, so that a wide range of augmentation structures can be represented. However, the user is still burdened with choosing the right augmentation structure. Instead, an alternative version of the LFR-based model augmentation structure has been proposed in~\cite{hoekstra_learning-based_2025}, utilizing continuous decision variables in the LFR matrix; hence, the final interconnection between the baseline and learning components is formed during model learning. The additionally introduced optimization variables lead to a highly flexible model structure, unifying the advantages of first-principles-based modeling and black-box identification. However, the price for the added generality is that an algebraic loop may appear in the model structure, i.e., the challenge of well-posedness arises. In general, there is no guarantee that a unique solution exists for the algebraic loop, which can ultimately lead to solver difficulties during model training. In~\cite{hoekstra_learning-based_2025}, structural constraints have been derived to achieve well-posed structures by triangularization of the feedback matrix. Although providing good experimental results, these constraints reduce the flexibility of the LFR-based formulation regarding the range of model augmentation structures that can be represented. 
Instead, we derive a direct model parametrization that guarantees well-posedness by construction, while also enabling the use of unconstrained gradient-based optimization algorithms, e.g., Adam \cite{kingma_adam_2015}.

Apart from the challenge of well-posedness, learning models with guaranteed stability properties is also a frequent problem in system identification. In general, identified models can be unstable, even if the data-generating system is stable. For linear system identification, several extensions are available for subspace methods with guaranteed stability, 
see, e.g.,~\cite{lacy_subspace_2003}. Stability plays an even more crucial role in deep learning-based nonlinear system identification. 
Due to the flexible nature of ANNs, \emph{Recurrent Neural Networks}~(RNNs), such as ANN-SS models, can be challenging to learn. During optimization, models can become unstable, causing the optimization to fail. 
This has motivated various research directions in recent years to develop ANN-based nonlinear system identification methods with stability guarantees. Most methods promote the stability property of RNNs either by hard constraints or regularization~\cite{terzi_learning_2021,de_carli_infinity-norm-based_2025}. Another promising direction for black-box ANN-SS models is based on direct parametrizations that guarantee stability by design, and hence allow for the use of unconstrained optimization~\cite{revay_recurrent_2024}. In this work, we develop such a parametrization approach for the LFR-based model augmentation structure while utilizing a highly efficient automatic differentiation tool, namely the JAX library~\cite{bradbury_jax_2018}, to facilitate rapid model learning.


The existing tools for guaranteeing well-posedness and stability of learned dynamic models are developed for black-box model structures. Extending these techniques to LFR-based interconnections of black-box and learning components is not straightforward; this motivates the proposed methods in this paper. The main contributions of this work are:\vspace{-2mm}
\begin{enumerate}
    \item We provide a direct, constraint-free parametrization which guarantees well-posedness of LFR-based model augmentation structures; moreover, we introduce an iterative approach for finding the fixed-point required for evaluating LFR-based models with theoretically guaranteed convergence.
    \item We propose a constraint-free parametrization for contracting LFR-based model augmentation to guarantee stability of the learned models.
    \item We adapt an identification pipeline based on~\cite{bemporad_l-bfgs-b_2025} with~$\ell_1$ and group-lasso penalties to enable model structure selection and discovery of the needed augmentation configuration.
    \item We demonstrate the effectiveness of our proposed model augmentation approach on various identification examples and benchmark problems.
\end{enumerate}\vspace{-2mm}
The paper is organized as follows: Section~\ref{sec:Augm_intro} introduces the model augmentation problem and LFR-based augmentation structure. It is followed by the proposed parametrization with the well-posedness guarantee in Section~\ref{sec:WellPosedness}, and the direct parametrization of the contracting LFR-based model augmentation structure in Section~\ref{sec:stable_param}. Then, Section~\ref{sec:ident_algorithm} presents the proposed identification pipeline with various regularization options, and the efficiency of the introduced model parametrization and estimation algorithm is demonstrated by multiple numerical experiments in Section~\ref{sec:num_examples}. Finally, conclusions are drawn in Section~\ref{sec:conclusion}.\vspace{-12pt}

%% file: sections/2_LFR_augmentation.tex
\section{LFR-based model augmentation setting}\label{sec:Augm_intro}\vspace{-6pt}
\subsection{Model augmentation problem statement}\label{sec:problem-statement}\vspace{-6pt}
Consider the dynamics of the data-generating system represented by the following \emph{discrete-time} (DT) nonlinear \emph{state-space} (SS) form:
\begin{subequations}
\label{eq:data-gen-sys}
{\setlength{\abovedisplayskip}{6pt}
 \setlength{\belowdisplayskip}{6pt}
\begin{align}
    x(k+1)&=f\left(x(k), u(k)\right),\\
    y(k)&=h\left(x(k), u(k)\right) + e(k),\label{eq:DT-y}
\end{align}
}\noindent
\end{subequations}
where $k \in \mathbb{Z}$ is the discrete time index, $x(k)\in\mathbb{R}^{n_\mathrm{x}}$ is the state, $u(k)\in\mathbb{R}^{n_\mathrm{u}}$ is the control input, $y(k)\in\mathbb{R}^{n_\mathrm{y}}$ is the measured output, while $f:\mathbb{R}^{n_\mathrm{x}} \times \mathbb{R}^{n_\mathrm{u}} \rightarrow \mathbb{R}^{n_\mathrm{x}}$ and $h:\mathbb{R}^{n_\mathrm{x}} \rightarrow \mathbb{R}^{n_\mathrm{y}}$ are possibly nonlinear functions. In~\eqref{eq:DT-y},~$e(k)$ represents the measurement noise generated by an i.i.d. white noise process with finite variance.

We assume that, based on prior knowledge and physical insight, a baseline model is available that describes some of the dominant dynamic relations of the system as
\begin{subequations}
\label{eqs:fp_model}
{\setlength{\abovedisplayskip}{6pt}
 \setlength{\belowdisplayskip}{6pt}
\begin{align}
    x_\mathrm{b}(k+1)&=f^\mathrm{FP}_{\theta_\mathrm{b}}(x_\mathrm{b}(k),u(k)),\\
    \hat{y}_\mathrm{b}(k)&=h^\mathrm{FP}_{\theta_\mathrm{b}}(x_\mathrm{b}(k), u(k)),
\end{align}
}\noindent
\end{subequations}
where $x_\mathrm{b}(k)\in\mathbb{R}^{n_\mathrm{x_b}}$ is the model state (or baseline state), $\hat{y}_\mathrm{b}(k)\in\mathbb{R}^{n_\mathrm{y}}$ is the baseline model output, $f^\mathrm{FP}_{\theta_\mathrm{b}}:\mathbb{R}^{n_{\mathrm{x_b}}} \times \mathbb{R}^{n_\mathrm{u}} \rightarrow \mathbb{R}^{n_{\mathrm{x_b}}}$ is the FP state transition function, and $h^\mathrm{FP}_{\theta_\mathrm{b}}:\mathbb{R}^{n_{\mathrm{x_b}}} \times \mathbb{R}^{n_\mathrm{u}} \rightarrow \mathbb{R}^{n_\mathrm{y}}$ is the output function, that both depend on $\theta_\mathrm{b}\in\mathbb{R}^{n_{\theta_\mathrm{b}}}$ physical (baseline) parameters, for which we assume $\theta_\mathrm{b}^0$ nominal values are known.

FP models are often subject to large amounts of uncertainty and unmodeled dynamics, which can be the main causes of model errors. To compensate for these effects, deep-learning-based model augmentation is introduced to enhance the model accuracy while still producing interpretable models. 
First, to capture unmodeled dynamic components, such as actuator dynamics and flexible modes, additional states are often required. Hence, the so-called \emph{augmentation states} are introduced in the following structure:
{\setlength{\abovedisplayskip}{6pt}
 \setlength{\belowdisplayskip}{6pt}
\begin{equation}
    \hat{x}(k)=\begin{bmatrix}
        x_\mathrm{b}^\top(k) & x_\mathrm{a}^\top(k)
    \end{bmatrix}^\top,
\end{equation}
}\noindent
where $x_\mathrm{a}(k) \in \mathbb{R}^{n_{\mathrm{x_a}}}$ denotes the additional state variables. Now $\hat{x}(k)\in\mathbb{R}^{n_{\mathrm{x_b}}+n_{\mathrm{x_a}}}$ stands for the combined model state of the augmentation structure. For systems where the unmodeled terms only affect the baseline model as missing static gains or nonlinearities, utilizing augmented states causes over-parametrization. Hence, it is a modeling decision when to use augmented states in the structure. In case $\mathrm{dim}\,x_\mathrm{a}=0$, we talk about \emph{static} augmentation, and when $\mathrm{dim}\,x_\mathrm{a}>0$, we call the corresponding model structure  \emph{dynamic} augmentation.

Model augmentation structurally combines a given first-principles model with learning components. 
A common practice is to jointly optimize the baseline model parameters $\theta_\mathrm{b}$ and the learning component parameters $\theta_\mathrm{a}\in\mathbb{R}^{n_{\theta_\mathrm{a}}}$ to achieve the best possible data-fit while tuning $\theta_\mathrm{b}$ as close to their physically true values as possible. Then, based on the \emph{Input-Output} (IO) data sequence $\mathcal{D}_N=\left\{\left(u_i,y_i\right)\right\}_{i=0}^{N-1}$ generated by~\eqref{eq:data-gen-sys}, the aim of model augmentation is to identify a model in the form of 
{\setlength{\abovedisplayskip}{6pt}
 \setlength{\belowdisplayskip}{6pt}
\begin{equation}
     \begin{bmatrix}
       \hat{x}^\top(k+1) & \hat{y}^\top(k)
    \end{bmatrix}^\top = \phi_\theta(\hat{x}(k), u(k)),
\end{equation}
}\noindent
where $\hat{y}(k)\in\mathbb{R}^{n_\mathrm{y}}$ is the output of the augmented model, $\theta=\mathrm{vec}\left(\theta_\mathrm{b},\theta_\mathrm{a}\right)$ combines the baseline and learning parameters, and  $\phi: \mathbb{R}^{n_{\mathrm{x_b}}+n_{\mathrm{x_a}}} \times \mathbb{R}^{n_\mathrm{u}} \rightarrow \mathbb{R}^{n_{\mathrm{x_b}}+n_{\mathrm{x_a}}} \times \mathbb{R}^{n_\mathrm{y}}$ determines both the state transition and the output map of the augmented model. The exact structure of $\phi_\theta$ depends on the applied model augmentation structure, e.g., the FP and learning components can be connected additively (parallel)~\cite{sun_comprehensive_2020}, multiplicatively (series)~\cite{gotte_composed_2022}, or various hybrid structures~\cite{hoekstra_learning-based_2025} are also available.\vspace{-6pt}

\subsection{LFR-based model augmentation structure}\label{sec:LFR-structure}\vspace{-6pt}
While a wide range of augmentation structures exists, selecting the most effective formulation is not straightforward, as it depends on both the data-generating system and the approximated baseline model. This challenge is addressed in~\cite{hoekstra_learning-based_2025}, where a generalized model augmentation structure based on LFRs is proposed with a parametrized interconnection between the baseline and learning components, allowing the structure itself to be tuned during model learning. Furthermore, another advantage of the LFR-based formulation is that there are various established control techniques available for LFRs (of uncertain systems), e.g.,~\cite{elghaoui_control_1996}. Thus, a significant advantage of the proposed structure is that the resulting models can be used directly for control design.

First, we introduce the following notation to denote the first-principles-based and the learning components, as
{\setlength{\abovedisplayskip}{6pt}
 \setlength{\belowdisplayskip}{6pt}
\begin{align}
    \phi^\mathrm{FP}_{\theta_\mathrm{b}}(z_\mathrm{b}(k)) &= \begin{bmatrix}
        f^\mathrm{FP}_{\theta_\mathrm{b}}(z_\mathrm{b}(k))\\
        h^\mathrm{FP}_{\theta_\mathrm{b}}(z_\mathrm{b}(k))
    \end{bmatrix},\\ 
    \phi^\mathrm{ANN}_{\theta_\mathrm{a}}(z_\mathrm{a}(k)) &= f^\mathrm{ANN}_{\theta_\mathrm{a}}(z_\mathrm{a}(k)),
\end{align}
}\noindent
where $f^\mathrm{ANN}$ is implemented as a fully-connected feedforward neural network, and $\theta_\mathrm{a}$ is collecting the network parameters such as weight and bias values, while $z_\mathrm{a}\in\mathbb{R}^{n_{\mathrm{z_a}}}$, $z_\mathrm{b}\in\mathbb{R}^{n_{\mathrm{x_b}}+n_\mathrm{u}}$ are latent variables, introduced for the LFR-based structure. The initial assumption is that $z_\mathrm{b}(k) = \mathrm{vec}(x_\mathrm{b}(k),u(k))$. Furthermore, latent variables $w_\mathrm{b}(k)\in\mathbb{R}^{n_{\mathrm{x_b}}+n_\mathrm{y}}$, and $w_\mathrm{a}(k)\in\mathbb{R}^{n_{\mathrm{w_a}}}$ are also introduced, and can be expressed as a function of the baseline and learning components, respectively:
{\setlength{\abovedisplayskip}{6pt}
 \setlength{\belowdisplayskip}{6pt}
\begin{equation}
    w(k) = \begin{bmatrix}
        w_\mathrm{b}(k)\\ w_\mathrm{a}(k)
    \end{bmatrix} = \underbrace{\begin{bmatrix}
        \phi^\mathrm{FP}_{\theta_\mathrm{b}}(z_\mathrm{b}(k))\\
        \phi^\mathrm{ANN}_{\theta_\mathrm{a}}(z_\mathrm{a}(k))
    \end{bmatrix}}_{\phi^\mathrm{NL}(z(k))},
\end{equation}
}\noindent
where $\phi^\mathrm{NL}$ contains the baseline and learning components. The latent variables $z_\mathrm{b}(k)$, $z_\mathrm{a}(k)$ are expressed as
{\setlength{\abovedisplayskip}{6pt}
 \setlength{\belowdisplayskip}{6pt}
\begin{multline}\label{eq:z(k)}
    \underbrace{\begin{bmatrix}
        z_\mathrm{b}(k)\\
        z_\mathrm{a}(k)
    \end{bmatrix}}_{z(k)} = \underbrace{\begin{bmatrix}
        C_\mathrm{z}^\mathrm{bb} & C_\mathrm{z}^\mathrm{ba}\\
        C_\mathrm{z}^\mathrm{ab} & C_\mathrm{z}^\mathrm{aa}
    \end{bmatrix}}_{C_\mathrm{z}} \underbrace{\begin{bmatrix}
        x_\mathrm{b}(k)\\ x_\mathrm{a}(k)
    \end{bmatrix}}_{\hat{x}(k)} + \underbrace{\begin{bmatrix}
        D_\mathrm{zu}^\mathrm{b}\\D_\mathrm{zu}^\mathrm{a}
    \end{bmatrix}}_{D_\mathrm{zu}} u(k) +\\
    \underbrace{\begin{bmatrix}
        D_\mathrm{zw}^\mathrm{bb} & D_\mathrm{zw}^\mathrm{ba}\\
        D_\mathrm{zw}^\mathrm{ab} & D_\mathrm{zw}^\mathrm{aa}
    \end{bmatrix}}_{D_\mathrm{zw}} \underbrace{\begin{bmatrix}
        w_\mathrm{b}(k)\\ w_\mathrm{a}(k)
    \end{bmatrix}}_{w(k)},
\end{multline}
}\noindent
where all elements of $C_\mathrm{z}$, $D_\mathrm{zu}$, and $D_\mathrm{zw}$ are optimized during model learning but can be handled separately, e.g., for parameter initialization purposes. The state transition and output map are expressed as
\begin{subequations}
{\setlength{\abovedisplayskip}{6pt}
 \setlength{\belowdisplayskip}{0pt}
\begin{multline}\label{eq:x(k+1)}
    \underbrace{\begin{bmatrix}
        x_\mathrm{b}(k+1)\\ x_\mathrm{a}(k+1)
    \end{bmatrix}}_{\hat{x}(k+1)} = \underbrace{\begin{bmatrix}
        A^\mathrm{bb} & A^\mathrm{ba}\\
        A^\mathrm{ab} & A^\mathrm{aa}
    \end{bmatrix}}_{A} \underbrace{\begin{bmatrix}
        x_\mathrm{b}(k)\\ x_\mathrm{a}(k)
    \end{bmatrix}}_{\hat{x}(k)} + \underbrace{\begin{bmatrix}
        B_\mathrm{u}^\mathrm{b}\\ B_\mathrm{u}^\mathrm{a}
    \end{bmatrix}}_{B_\mathrm{u}} u(k) +\\
     \underbrace{\begin{bmatrix}
        B_\mathrm{w}^\mathrm{bb} & B_\mathrm{w}^\mathrm{ba}\\
        B_\mathrm{w}^\mathrm{ab} & B_\mathrm{w}^\mathrm{aa}
    \end{bmatrix}}_{B_\mathrm{w}} \underbrace{\begin{bmatrix}
        w_\mathrm{b}(k)\\ w_\mathrm{a}(k)
    \end{bmatrix}}_{w(k)},
\end{multline}}
{\setlength{\abovedisplayskip}{0pt}
 \setlength{\belowdisplayskip}{6pt}
\begin{multline}\label{eq:y(k)}
    \hat{y}(k) = \underbrace{\begin{bmatrix}
        C_\mathrm{y}^\mathrm{b} & C_\mathrm{y}^\mathrm{a}
    \end{bmatrix}}_{C_\mathrm{y}} \underbrace{\begin{bmatrix}
        x_\mathrm{b}(k)\\ x_\mathrm{a}(k)
    \end{bmatrix}}_{\hat{x}(k)} + D_\mathrm{yu} u(k) +\\ \underbrace{\begin{bmatrix}
        D_\mathrm{yw}^\mathrm{b} & D_\mathrm{yw}^\mathrm{a}
    \end{bmatrix}}_{D_\mathrm{yw}} \underbrace{\begin{bmatrix}
        w_\mathrm{b}(k)\\ w_\mathrm{a}(k)
    \end{bmatrix}}_{w(k)},
\end{multline}
}\noindent
\end{subequations}
where $A$, $B_\mathrm{u}$, $C_\mathrm{y}$, and $D_\mathrm{yu}$ matrices represent the linear parts of the unmodeled dynamics, the baseline part is introduced through $B_\mathrm{w}^\mathrm{bb}$, $B_\mathrm{w}^\mathrm{ba}$, and $D_\mathrm{yw}^\mathrm{b}$, while the nonlinear learning component terms are considered with $B_\mathrm{w}^\mathrm{ab}$, $B_\mathrm{w}^\mathrm{aa}$, and $D_\mathrm{yw}^\mathrm{a}$. Finally, the LFR-based structure can be expressed in a compact form, as
\begin{subequations}\label{eqs:general_LFR}
\begin{align}
    \begin{bmatrix}
        \hat{x}(k+1)\\ \hat{y}(k)\\ z(k)
    \end{bmatrix} &= \underbrace{\begin{bmatrix}
        A & B_\mathrm{u} & B_\mathrm{w}\\
        C_\mathrm{y} & D_\mathrm{yu} & D_\mathrm{yw}\\
        C_\mathrm{z} & D_\mathrm{zu} & D_\mathrm{zw}
    \end{bmatrix}}_{W_{\theta_\mathrm{LFR}}} \begin{bmatrix}
        \hat{x}(k)\\ u(k)\\ w(k)
    \end{bmatrix},\\
    w(k) &= \begin{bmatrix}
        \phi^\mathrm{FP}_{\theta_\mathrm{b}}(z_\mathrm{b}(k))\\
        \phi^\mathrm{ANN}_{\theta_\mathrm{a}}(z_\mathrm{a}(k))
    \end{bmatrix},
\end{align}
\end{subequations}
where $W_{\theta_\mathrm{LFR}}$ is the LFR matrix, with the optimization parameters describing the interconnection structure collected into the $\theta_\mathrm{LFR}$ parameter vector, i.e., $\theta_\mathrm{LFR}=\mathrm{vec}(A, B_\mathrm{u},\dots,D_\mathrm{zw})$. All sub-matrices in $W_{\theta_\mathrm{LFR}}$ are tuned during model training, hence the generality of the approach. All available model augmentation structures can be represented by assigning specific values to the LFR matrices, as shown in~\cite{hoekstra_learning-based_2025}. The LFR-based model augmentation structure is illustrated in Fig.~\ref{fig:LFR_augm_struct}.\vspace{-6pt}
\begin{figure}
    \centering
    \includegraphics[scale=0.85]{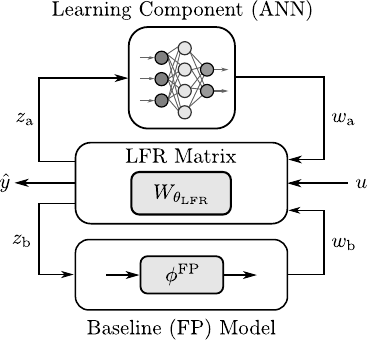}\vspace{-6pt}
    \caption{LFR-based model augmentation structure.}
    \label{fig:LFR_augm_struct}
\end{figure}

\subsection{Identification setting}\vspace{-6pt}
The standard approach for estimating \emph{Output-Error}~(OE) model structures such as \eqref{eqs:general_LFR} from a data set $\mathcal{D}_N$ is to minimize the simulation error, expressed as
\begin{equation}\label{eq:cost-fun}
    V_{\mathcal{D}_N}(\theta, \hat{x}_0) = \dfrac{1}{N} \sum_{k=0}^{N-1} \left\|y(k) - \hat{y}(k)\right\|_2^2,
\end{equation}
where $y(k)$ denotes the measured output and $\hat{y}(k)$ the simulated output obtained from \eqref{eqs:general_LFR}. Simulation error minimization is a well-established method in nonlinear system identification and is known to produce highly accurate models~\cite{schoukens_nonlinear_2019}. However, gradient-based algorithms often encounter stability issues when minimizing~\eqref{eq:cost-fun}, see~\cite{ribeiro_smoothness_2020}. To address this challenge,~\cite{hoekstra_learning-based_2025} adapted the truncated objective function (see, e.g.,~\cite{beintema_deep_2023}) for estimating model augmentation structures, and employed an encoder network to estimate the initial state of each truncated subsequence. While this strategy achieved excellent accuracy across a range of identification problems~\cite{hoekstra_learning-based_ECC_2025,gyorok_orthogonal_2025}, the selection of the truncation length remains a user-dependent choice that can influence model performance. Moreover, the use of multiple overlapping subsequences often leads to increased computational time in practice. 
In contrast, we follow a more conventional approach by treating the initial state $\hat{x}(0) \coloneq \hat{x}_0$ as an optimization variable, and estimate the combined parameter vector $\theta = \mathrm{vec}(\theta_\mathrm{b}, \theta_\mathrm{a},\theta_\mathrm{LFR})$ by solving
\begin{equation}\label{eq:optim-problem}
\begin{aligned}
\min_{\theta, \hat{x}_0} \quad & V_{\mathcal{D}_N}(\theta, \hat{x}_0),\\
\textrm{s.t.} \quad & \begin{bmatrix}
    \hat{x}(k+1)\\\hat{y}(k)\\z(k)
\end{bmatrix} = W_{\theta_\mathrm{LFR}} \begin{bmatrix}
    \hat{x}(k)\\ u(k)\\w(k)
\end{bmatrix},\\
    &  w(k) = \phi^\mathrm{NL}(z(k)),\quad k\in\left[0,N-1\right].
\end{aligned}
\end{equation}
The formulation in \eqref{eq:optim-problem} facilitates faster model training compared to the method in~\cite{hoekstra_learning-based_2025}. However, several challenges must be addressed. First, the algebraic loop in the LFR structure \eqref{eqs:general_LFR} must be handled carefully to ensure the existence and uniqueness of solutions. This is achieved through the well-posed parametrization introduced in Section~\ref{sec:WellPosedness}. Second, simulation-error minimization may lead to unstable models. Even though cost function~\eqref{eq:cost-fun} strongly penalizes models that become unstable during training, the estimated model might still become unstable on previously unseen data due to the nonlinear nature of the model. This issue is resolved by the stable-by-construction parametrization presented in Section~\ref{sec:stable_param}. Finally, backpropagation through long datasets using the simulation error loss can be computationally demanding, and the associated optimization problem is prone to convergence to local minima. We mitigate these issues by adopting the computationally efficient identification pipeline of~\cite{bemporad_efficient_2025}, leveraging the L-BFGS-B optimizer for its rapid convergence to high-quality solutions. Another major advantage of the applied identification setting is that it can handle $\ell_1$ and group-lasso regularization, enabling automatic augmentation structure discovery and model order selection, as we demonstrate in Section~\ref{sec:ident_algorithm}.\vspace{-6pt}

%% file: sections/3_WP_parametrization.tex
\section{Well-posedness of the LFR-based structure}\label{sec:WellPosedness}\vspace{-6pt}
\subsection{Fixed-point iterations}\vspace{-6pt}
As mentioned in Section~\ref{sec:Augm_intro}, the LFR-based model augmentation structure \eqref{eqs:general_LFR} contains an algebraic loop, i.e., 
\begin{equation}\label{eq:well-posedness-equation}
\left\{
\begin{aligned}
    z(k) &= C_\mathrm{z} \hat{x}(k) + D_\mathrm{zu} u(k) + D_\mathrm{zw} w(k),\\
    w(k) &= \phi^\mathrm{NL}(z(k)),
\end{aligned} \right.
\end{equation}
set of equations needs to be solved at every time step to evaluate the model. To address this issue, \cite{hoekstra_learning-based_2025} proposed constraints on the $D_\mathrm{zw}$ matrix that guarantee well-posedness, i.e., to guarantee that \eqref{eq:well-posedness-equation} has a unique solution $\forall \hat{x}(k)\in\mathbb{R}^{n_{\mathrm{x_b}}+n_{\mathrm{x_a}}}$, $\forall u(k)\in\mathbb{R}^{n_\mathrm{u}}$. For example, restricting $D_\mathrm{zw}\equiv 0$ resolves the algebraic loop in \eqref{eq:well-posedness-equation}, however also limits the number of model augmentation structures that can be represented by the LFR-based formulation. Another option is to select either $D_\mathrm{zw}^\mathrm{ab}$ or~$D_\mathrm{zw}^\mathrm{ba}$ to be tuned freely and set all other members of~$D_\mathrm{zw}$ to zero. The latter approach is more general; however, it still applies restrictions to the interconnection matrix, which limits the flexibility of the structure. Furthermore, it is not intuitive whether to select $D_\mathrm{zw}^\mathrm{ab}$ or $D_\mathrm{zw}^\mathrm{ba}$ to be optimized. Therefore, we aim to guarantee the well-posedness of the LFR-based model augmentation structure by parametrization. To do that, first,~\eqref{eq:well-posedness-equation} is rearranged as
\begin{equation}\label{eq:well-posedness_zk}
    z(k) = \underbrace{D_\mathrm{zw} \phi^\mathrm{NL}(z(k)) + C_\mathrm{z} x(k) + D_\mathrm{zu} u(k)}_{g(z(k), x(k), u(k))}.
\end{equation}
We now formulate a condition for $g(\cdot)$ to guarantee well-posedness of the LFR-based augmentation structure by using the following definition. \vspace{-3pt}
\begin{defn}\label{def:contraction_mapping}
    The function $g(\cdot):\mathbb{R}^{n_\mathrm{z}} \times \mathbb{R}^{n_\mathrm{x}} \times \mathbb{R}^{n_\mathrm{u}}\rightarrow \mathbb{R}^{n_\mathrm{z}}$ in \eqref{eq:well-posedness_zk} is a contraction mapping if there exists $c\in\left[0,1\right)$ for any $(x,u)\in\mathbb{R}^{n_\mathrm{x}} \times \mathbb{R}^{n_\mathrm{u}}$ such that
    \begin{equation}
        \|g(z_1, x, u) - g(z_1, x, u)\|_2 \leq c\|z_1-z_2\|_2,\quad \forall z_1, z_2 \in \mathbb{R}^{n_\mathrm{z}}.
    \end{equation}
\end{defn}\vspace{-3pt}
\begin{prop}\label{prop:FPI}
    If $g(\cdot)$ in \eqref{eq:well-posedness_zk} is a contraction mapping, by the Banach theorem, there exists a unique $z_\ast$ solution to $z=g(z,x,u)$ for any $(x,u)\in\mathbb{R}^{n_\mathrm{x}} \times \mathbb{R}^{n_\mathrm{u}}$, i.e., the LFR-based model augmentation structure \eqref{eqs:general_LFR} is well-posed. Furthermore, this solution can be found starting from any $z_0 \in \mathbb{R}^{n_\mathrm{z}}$ with fixed-point iterations, as
\begin{subequations}\label{eqs:FPI}
\begin{align}
    z_n &= g(z_{n-1}, x, u), \quad \forall n\in\mathbb{Z}_{>0},\\
    z_\ast &= \lim_{n\to \infty} z_n.
\end{align}
\end{subequations}
\end{prop}\vspace{-12pt}
According to Proposition~\ref{prop:FPI}, if it is guaranteed by the parametrization of $g$ that $g$ is a contraction mapping, then the model augmentation structure is well-posed, and for any $x(k)$, $u(k)$ value at time moment $k$ the fixed-point $z_\ast(k)$ can be found iteratively. To achieve this, certain conditions for the feedback function $\phi^\mathrm{NL}$ need to be satisfied, both for the baseline and learning parts.\vspace{-6pt}

\subsection{Lipschitz continuous baseline model}\vspace{-6pt}
For the general $D_\mathrm{zw} \neq 0$ case, $g(\cdot)$ being a contraction mapping implies that $\phi^\mathrm{NL}$ is Lipschitz bounded. To achieve this condition, first, the following assumption is made regarding the baseline part.\vspace{-3pt}
\begin{assum}\label{assumption:baseline-lipschitz}
    The baseline state transition function $f^\mathrm{FP}_{\theta_\mathrm{b}}$ and output map $h^\mathrm{FP}_{\theta_\mathrm{b}}$ are both locally Lipschitz continuous functions on a known operating range given by $\mathbb{X}\subset \mathbb{R}^{n_{\mathrm{x_b}}}$ and $\mathbb{U}\subset \mathbb{R}^{n_\mathrm{u}}$ such that $x_\mathrm{b}(k) \in \mathbb{X}$, $u(k) \in \mathbb{U}$. Furthermore, $L_f$ and $L_h$, the Lipschitz constants of $f^\mathrm{FP}_{\theta_\mathrm{b}}$ and $h^\mathrm{FP}_{\theta_\mathrm{b}}$, respectively, are known constants.
\end{assum}\vspace{-3pt}
For \emph{Linear Time-Invariant} (LTI) baseline models, Assumption~\ref{assumption:baseline-lipschitz} trivially holds globally; moreover, $L_f$ and $L_h$ Lipschitz constants can be found as shown in Lemma~\ref{lemma:LTI-Lipschitz}.\vspace{-3pt}
\begin{lem}\label{lemma:LTI-Lipschitz}
    Given an LTI-SS baseline model, as
    \begin{subequations}
    \begin{align}
        f^\mathrm{FP}_{\theta_\mathrm{b}}(x_\mathrm{b}(k), u(k)) &= A^\mathrm{FP}_{\theta_\mathrm{b}} x_\mathrm{b}(k) + B^\mathrm{FP}_{\theta_\mathrm{b}} u(k),\\
        h^\mathrm{FP}_{\theta_\mathrm{b}}(x_\mathrm{b}(k), u(k)) &= C^\mathrm{FP}_{\theta_\mathrm{b}} x_\mathrm{b}(k) + D^\mathrm{FP}_{\theta_\mathrm{b}} u(k),
    \end{align}
    \end{subequations}
    Lipschitz constants $L_f$ and $L_h$ can be calculated as
    \begin{equation}
        L_f = \left\| \begin{bmatrix}
            A^\mathrm{FP}_{\theta_\mathrm{b}}  B^\mathrm{FP}_{\theta_\mathrm{b}}
        \end{bmatrix} \right\|_2,\quad
        L_h = \left\| \begin{bmatrix}
            C^\mathrm{FP}_{\theta_\mathrm{b}}  D^\mathrm{FP}_{\theta_\mathrm{b}}
        \end{bmatrix} \right\|_2,
    \end{equation}
    where $\|\cdot\|_2$ denotes the spectral norm.
\end{lem}\vspace{-14pt}
\begin{pf}
    See Appendix~\ref{appendix:lemma1}.
\end{pf}\vspace{-12pt}
For nonlinear baseline models, a possible way of computing the Lipschitz constant is to find the upper bound of their derivatives inside the region $x_\mathrm{b}(k) \in \mathbb{X}$, $u(k) \in \mathbb{U}$. Since it is reasonable to assume that both $f^\mathrm{FP}_{\theta_\mathrm{b}}$ and $h^\mathrm{FP}_{\theta_\mathrm{b}}$ are (at least once) continuously differentiable functions, the Lipschitz constant can be computed as
\begin{equation}
    \mathrm{Lip}\left\{f^\mathrm{FP}_{\theta_\mathrm{b}}\right\} = \sup_{x_\mathrm{b}\in\mathbb{X}, u\in\mathbb{U}} \left\| \pdv{f^\mathrm{FP}_{\theta_\mathrm{b}}}{z_\mathrm{b}}\right\|_2.
\end{equation}
Here, $z_\mathrm{b} = \mathrm{vec}(x_\mathrm{b}, u)$ and the time notation is omitted for convenience.\vspace{-3pt}
\begin{rem}
    If the baseline parameters $\theta_\mathrm{b}$ are tuned during the optimization, the Lipschitz constants that are calculated using the nominal parameter values $\theta_\mathrm{b}^0$ might change. To prevent possible issues, it is advised to overestimate the Lipschitz bounds. Moreover, applying regularization (see Section~\ref{sec:phys-parm-regul}) may limit the deviation of the baseline parameters from their nominal values, thus preserving the accuracy of the original Lipschitz constants.
\end{rem}\vspace{-6pt}
\subsection{Lipschitz constant of the learning component}\label{sec:ANN-lipschitz}\vspace{-6pt}
Based on Assumption~\ref{assumption:baseline-lipschitz}, $f^\mathrm{FP}_{\theta_\mathrm{b}}$ is $L_f$-Lipschitz, while $h^\mathrm{FP}_{\theta_\mathrm{b}}$ is $L_h$-Lipschitz. Then it is reasonable to construct the learning component according to Condition~\ref{cond:ANN-lipschitz}.\vspace{-3pt}
\begin{cond}\label{cond:ANN-lipschitz}
    The learning component $\phi^\mathrm{ANN}_{\theta_\mathrm{a}}$ in \eqref{eqs:general_LFR} is $L_\mathrm{ANN}$-Lipschitz with $L_\mathrm{ANN}{\leq} L$, where $L=\max \{L_f,\,L_h\}$.
\end{cond}\vspace{-3pt}
Condition~\ref{cond:ANN-lipschitz} can be enforced by following different strategies. The most common approach is to guarantee that each layer in the ANN is 1-Lipschitz, then multiply the output by~$L$. ANN layers with a Lipschitz constant of 1 can be enforced, e.g., by layer normalization~\cite{miyato_spectral_2018}, or projected gradient-based training~\cite{gouk_regularisation_2021}. 
Recently, various research also aimed to guarantee Lipschitz bounded ANNs directly by parametrization, see, e.g.,~\cite{wang_direct_2023}. In this paper, we take a relatively simple approach that constrains the Lipschitz constant of the learning component by regularization. The learning component is considered as a feedforward ANN with~$H$ hidden layers. For the applied activation function~$\sigma$, we make the following assumption.\vspace{-3pt}
\begin{assum}\label{assumption:lipschitz-act-fun}
    The activation function $\sigma: \mathbb{R} \rightarrow \mathbb{R}$ is 1-Lipschitz.
\end{assum}\vspace{-3pt}
Note that Assumption~\ref{assumption:lipschitz-act-fun} is typically true for the most commonly applied activation functions~\cite{goodfellow_deep_2016}. Since computing the exact Lipschitz constant is an NP-hard problem even for a neural network with 2 hidden layers~\cite{virmaux_lipschitz_2018}, we utilize the following, commonly applied approximation for the upper bound for $L_\mathrm{ANN}$ \cite{gouk_regularisation_2021}:
{\setlength{\abovedisplayskip}{6pt}
 \setlength{\belowdisplayskip}{6pt}
\begin{equation}\label{eq:lipschitz-ann-upper}
    \mathrm{Lip}\left\{\phi_\mathrm{\theta_\mathrm{a}}^\mathrm{ANN}\right\} = L_\mathrm{ANN} \leq\prod_{i=0}^H \left\|W_i\right\|_2.
\end{equation}
}\noindent
Based on \eqref{eq:lipschitz-ann-upper}, a regularization term is added to the cost function \eqref{eq:cost-fun} for bounding the Lipschitz constant of the learning component:
{\setlength{\abovedisplayskip}{6pt}
 \setlength{\belowdisplayskip}{6pt}
\begin{equation}\label{eq:ANN_lipschitz_regul}
    r_\mathrm{L}(\theta_\mathrm{a}) = \rho_\mathrm{L} \max\left\{\prod_{i=0}^H \left\|W_i\right\|_2 - L, 0\right\}^2,
\end{equation}
}\noindent
where $\rho_\mathrm{L} \gg 1$ is a trade-off parameter. If the optimization results in such $\theta_\mathrm{a}$ parameters for which Condition~\ref{cond:ANN-lipschitz} is not satisfied, model training can be repeated with an increased $\rho_\mathrm{L}$ value. Note that the presented approach can be interpreted as a soft-constraint for the Lipschitz constant of the learning component; however, any of the previously mentioned methods can be incorporated into the proposed LFR-based model augmentation setting for a stricter guarantee of fulfilling Condition~\ref{cond:ANN-lipschitz}.\vspace{-6pt}

\subsection{Parametrizing the $D_\mathrm{zw}$ matrix}\vspace{-6pt}
With Assumption~\ref{assumption:baseline-lipschitz} holding and Condition~\ref{cond:ANN-lipschitz} satisfied, it is guaranteed that $\phi^\mathrm{NL}$ is $L$-Lipschitz. We now provide a direct (i.e., constraint-free) parametrization for the $D_\mathrm{zw}\in\mathbb{R}^{n_\mathrm{z}\times n_\mathrm{w}}$ matrix to guarantee that $\|D_\mathrm{zw}\|_2<1/L$. This condition ensures that $g(\cdot)$ is a contraction mapping, according to Definition~\ref{def:contraction_mapping}. For the proposed parametrization, the Cayley transform is utilized. Given a square matrix $M$ with $I+M$ invertible, its Cayley transform is defined as $\mathrm{Cayley}(M)=(I-M)(I+M)^{-1}$. For non-square matrices, it can be generalized as follows:\vspace{-3pt}
\begin{lem}[Generalized Cayley transform \cite{verhoek_learning_2023}]\label{lemma:general_Cayley}
    Let $M\in\mathbb{R}^{n\times m}$ with $n \geq m$. Then $M^\top M \prec I$ if and only if there exist $X, Y \in \mathbb{R}^{m\times m}$ and $Z\in\mathbb{R}^{(n-m)\times m}$ such that
    \begin{equation}\label{eq:general_cayley}
        M = \begin{bmatrix}
            \mathrm{Cayley}(N)\\
            -2Z(I+N)^{-1}
        \end{bmatrix},
    \end{equation}
    where $N=X^\top X + (Y - Y^\top) + Z^\top Z + \epsilon I$, with $0<\epsilon \ll 1$.
\end{lem}\vspace{-3pt}
To parametrize the $D_\mathrm{zw}$ matrix, we introduce the following free variables: $X_\mathrm{D}, Y_\mathrm{D}\in\mathrm{R}^{n\times n}$, $Z_\mathrm{D}\in\mathbb{R}^{(n-m)\times m}$, and $d\in\mathbb{R}$, where $n=\mathrm{max}(n_\mathrm{z}, n_\mathrm{w})$, and similarly $m=\mathrm{min}(n_\mathrm{z}, n_\mathrm{w})$. Then, the $D_\mathrm{zw}$ matrix is expressed as
{\setlength{\abovedisplayskip}{6pt}
 \setlength{\belowdisplayskip}{6pt}
\begin{equation}\label{eq:Dzw-final-parametrization}
    D_\mathrm{zw} = \begin{cases}
        \frac{\sigma_\mathrm{D}}{L} \bar{D}_\mathrm{zw}, &\text{if $n_\mathrm{w}\geq n_\mathrm{z}$},\\
        \frac{\sigma_\mathrm{D}}{L}\bar{D}_\mathrm{zw}^\top, &\text{if $n_\mathrm{w} > n_\mathrm{z}$},
    \end{cases}
\end{equation}
}\noindent
where scaling factor $\sigma_\mathrm{D}$ is computed\footnote{\scriptsize By the sigmoid function, we refer to $\mathrm{sigmoid}(x)=\frac{1}{1+e^{-x}}$. We assume that $x$ remains bounded for bounded signals, 
hence $\mathrm{sigmoid}(x){\in}(0,1)$. This boundedness assumption can be replaced by using $1{-}\epsilon$ in the numerator of the sigmoid function, where $\epsilon$ is an arbitrary small positive value.} as $\sigma_D = \mathrm{sigmoid}(d)$, and $\bar{D}_\mathrm{zw}$ is computed as follows:
{\setlength{\abovedisplayskip}{6pt}
 \setlength{\belowdisplayskip}{6pt}
\begin{equation}\label{eq:barD_zw_param}
    \bar{D}_\mathrm{zw} = \begin{bmatrix}
        \mathrm{Cayley}(N_\mathrm{D})\\
        -2Z_\mathrm{D}(I + N_\mathrm{D})^{-1}
    \end{bmatrix},
\end{equation}
}\noindent
where $N_\mathrm{D}=X_\mathrm{D}^\top X_\mathrm{D}+Y_\mathrm{D}-Y_\mathrm{D}^\top + Z_\mathrm{D}^\top Z_\mathrm{D}+\epsilon I_n$. Here $\epsilon\in\mathbb{R}_{>0}$ is a small user-specified constant. According to Lemma~\ref{lemma:general_Cayley}, \eqref{eq:barD_zw_param} ensures that $\|\bar{D}_\mathrm{zw}\|_2<1$. Finally, the right-hand side of 
\eqref{eq:Dzw-final-parametrization} guarantees that $\|D_\mathrm{zw}\|_2<1/L$ by parametrization, in a constraint-free manner.\vspace{-3pt}
%
%
\begin{rem}
    Instead of the generalized Cayley transform, a matrix-exponential-based parametrization can be used to ensure well-posedness; see, e.g., \cite{drenth_efficient_2025}. However, the Cayley transform-based parametrization can be applied for non-square matrices, which is crucial for LFR augmentation structures, as often  $n_\mathrm{z}\neq n_\mathrm{w}$. 
\end{rem}\vspace{-6pt}

\subsection{Resulting contraction mapping}\vspace{-6pt}
Finally, with the proposed parametrization and conditions regarding the baseline and learning components, the well-posedness of the LFR-based model augmentation function can be proven.\vspace{-3pt}
\begin{thm}\label{thm:WP}
    Let Assumption~\ref{assumption:baseline-lipschitz} and Condition~\ref{cond:ANN-lipschitz} hold. If $D_\mathrm{zw}$ is parametrized as outlined in \eqref{eq:Dzw-final-parametrization}, then the LFR-based model augmentation structure \eqref{eqs:general_LFR} is well-posed.
\end{thm}\vspace{-12pt}
\begin{pf}
    According to Proposition~\ref{prop:FPI}, \eqref{eqs:general_LFR} is well-posed if $g(\cdot)$ in \eqref{eq:well-posedness_zk} is a contraction mapping. 
    Hence, by substituting \eqref{eq:well-posedness_zk} into Definition~\ref{def:contraction_mapping}, there should exist a $c\in\left[0,\,1\right)$ for guaranteed well-posedness  with the following property:
    \begin{multline}
        \|D_\mathrm{zw}\phi^\mathrm{NL}(z_1) + C_\mathrm{z} x + D_\mathrm{zu} u - D_\mathrm{zw}\phi^\mathrm{NL}(z_2) - C_\mathrm{z} x - D_\mathrm{zu} u\|_2\\ \leq c\norm{z_1 - z_2}_2,
    \end{multline}
    for any $x\in\mathbb{R}^{n_\mathrm{x}}$, $u\in\mathbb{R}^{n_\mathrm{u}}$, and $z_1, z_2\in\mathbb{R}^{n_\mathrm{z}}$. With rearranging and dropping out terms, the following inequality needs to hold for any $c\in\left[0,\,1\right)$:
    \begin{equation}\label{eq:contraction_ineq}
        \norm{D_\mathrm{zw}}_2 \norm{\phi^\mathrm{NL}(z_1) - \phi^\mathrm{NL}(z_2)}_2 \leq c\norm{z_1 - z_2}_2.
    \end{equation}
    Since $\|D_\mathrm{zw}\|_2<1/L$ by parametrization, and according to Assumption~\ref{assumption:baseline-lipschitz} and Condition~\ref{cond:ANN-lipschitz}, $\mathrm{Lip}\{\phi^\mathrm{NL}\}\leq L$, i.e.,
    \begin{equation}\label{eq:phi_NL_lipschitz}
        \dfrac{\norm{\phi^\mathrm{NL}(z_1) - \phi^\mathrm{NL}(z_2)}_2}{\norm{z_1 - z_2}_2}  \leq L,
    \end{equation}
    inequality \eqref{eq:contraction_ineq} is fulfilled.\placeqed
\end{pf}\vspace{-12pt}
According to Theorem~\ref{thm:WP}, \eqref{eq:Dzw-final-parametrization} provides a constraint-free parametrization to guarantee well-posedness of the LFR augmentation structure (under Condition~\ref{cond:ANN-lipschitz}), by tuning the free variables $X_\mathrm{D}, Y_\mathrm{D}, Z_\mathrm{D}$, and $d$ instead of directly estimating the elements of $D_\mathrm{zw}$. Consequently, optimization problem~\eqref{eq:optim-problem} can be solved with an unconstrained algorithm, such as \emph{stochastic gradient descent}~(SGD). This concludes Contribution~1 of the paper.\vspace{-6pt}

%% file: sections/4_stability.tex
\section{Stable-by-construction model augmentation}\label{sec:stable_param}\vspace{-6pt}
Having established well-posedness of the LFR-based augmentation structure through the unconstrained parametrization presented in Section~\ref{sec:WellPosedness}, we now turn to the problem of constructing a direct parametrization that guarantees stability. A key requirement for the proposed method is to be compatible with the well-posedness approach to enable enforcing both properties simultaneously with a unified parametrization.\vspace{-6pt}

\subsection{Contracting model property}\vspace{-6pt}
As discussed in Section~\ref{sec:intro}, in recent years, various approaches have been proposed to guarantee stability of the model estimate in nonlinear system identification. Many of these methods, e.g.,~\cite{revay_recurrent_2024,verhoek_learning_2023}, utilize a special form of nonlinear stability, namely contraction, since it is especially suitable for model learning applications.\vspace{-3pt}
\begin{defn}\label{def:contracting_model}
    The model represented by \eqref{eqs:general_LFR} is said to be contracting with rate $\alpha\in(0,\,1)$ if for any two initial conditions $\hat{x}_\mathrm{i}(0),\,\hat{x}_\mathrm{j}(0) \in\mathbb{R}^{n_{\hat{x}}}$, given the same, bounded input sequence $\{u(k)\}_{k=0}^{k\to \infty}$, the corresponding state sequences $\hat{x}_\mathrm{i}$, $\hat{x}_\mathrm{j}$ satisfy
    \begin{equation}
        \|\hat{x}_\mathrm{i}(k) - \hat{x}_\mathrm{j}(k)\|_2 \leq K\alpha^k \|\hat{x}_\mathrm{i}(0) - \hat{x}_\mathrm{j}(0)\|_2, \quad \forall k>0,
    \end{equation}
    for some $K>0$.
\end{defn}\vspace{-3pt}
From Definition~\ref{def:contracting_model}, the benefits of the contraction property in system identification become clear. 
In general, small perturbations to the initial condition may lead to drastically different model responses. In the context of model training, this sensitivity can cause the optimization to get stuck in certain regimes or even to diverge. Enforcing contraction of the estimated model provides an effective solution to this challenge. Moreover, if the true system is stable, it should be reflected in the identified model. However, naive simulation-error minimization may yield unstable models, even if the underlying data-generating system is stable. This issue can also be avoided by utilizing a stable-by-construction parametrization. Finally, contraction (or incremental stability) of the predictor is a key requirement for proving statistical consistency of the estimator (see, e.g.,~\cite{beintema_deep_2023}). While this property could be verified a posteriori, the proposed contracting parametrization ensures that the stability condition is satisfied by design.\vspace{-6pt}

\subsection{Conditions for contracting parametrization}\vspace{-6pt}
To derive the necessary conditions for the contracting property, we consider the incremental dynamics, i.e., the error dynamics between two trajectories $(\hat{x}_i, z_i, w_i)$, and $(\hat{x}_j, z_j, w_j)$ under the same input sequence. The incremental form of \eqref{eqs:general_LFR} can be expressed as
\begin{subequations}\label{eqs:incremental_dyn}
\begin{align}
    \begin{bmatrix}
        \delta_{\hat{x}}(k+1)\\ \delta_{\hat{y}}(k)\\ \delta_z(k)
    \end{bmatrix} &= \begin{bmatrix}
        A & B_\mathrm{u} & B_\mathrm{w}\\
        C_\mathrm{y} & D_\mathrm{yu} & D_\mathrm{yw}\\
        C_\mathrm{z} & D_\mathrm{{zu}} & D_\mathrm{{zw}}
    \end{bmatrix} \begin{bmatrix}
        \delta_{\hat{x}}(k)\\ \delta_u(k)\\ \delta_w(k)
    \end{bmatrix},\\
    \delta_w(k) &= \phi^\mathrm{NL}(z_j(k) + \delta_z(k)) - \phi^\mathrm{NL}(z_j(k)),\label{eq:incremental_dyn_NL}
\end{align}
\end{subequations}
where $\delta_{\hat{x}} = \hat{x}_i - \hat{x}_j$, and $\delta_{\hat{y}}$, \dots, $\delta_w$ are defined similarly.

Since $\phi^\mathrm{NL}$ is Lipschitz bounded, it is continuously differentiable w.r.t. $z(k)$, and the fundamental theorem of calculus~\cite{koelewijn_analysis_2023} can be applied to provide a factorized form:
\begin{multline}\label{eq:Phi_NL_FTC}
    \phi^\mathrm{NL}(z_j(k) + \delta_z(k)) - \phi^\mathrm{NL}(z_j(k)) =\\ \underbrace{\left(\int_0^1 \frac{\partial\phi^\mathrm{NL}}{\partial z} (z_j(k) + \lambda \delta_z(k))\,\mathrm{d}\lambda \right)}_{P^\mathrm{NL}(z_\mathrm{i}(k), z_\mathrm{j}(k))} \delta_z(k).
\end{multline}
The contracting property is defined for trajectories with the same input sequence (i.e., $\delta_u \equiv 0$), substituting \eqref{eq:Phi_NL_FTC} back into \eqref{eqs:incremental_dyn} gives:
\begin{multline}\label{eq:delta_w_expression}
    \delta_w(k) = \left( I_{n_\mathrm{w}} -P^\mathrm{NL}(z_\mathrm{i}(k), z_\mathrm{j}(k)) D_\mathrm{zw}\right)^\mathrm{-1} \\P^\mathrm{NL}(z_\mathrm{i}(k), z_\mathrm{j}(k)) C_\mathrm{z} \delta_x(k),
\end{multline}
where $( I_{n_w} - P^\mathrm{NL}(\cdot) D_\mathrm{zw})^\mathrm{-1}$ exists by Theorem~\ref{thm:WP}.
%
%
Then, the incremental state transition function can be expressed as
\begin{multline}\label{eq:delta_x_map}
    \delta_x(k+1) =\\ \underbrace{\left(A + B_\mathrm{w}\left( I_{n_w} - P^\mathrm{NL}(\cdot) D_\mathrm{zw}\right)^\mathrm{-1} P^\mathrm{NL}(\cdot) C_\mathrm{z}\right)}_{\mathcal{A}(z_\mathrm{i}(k), z_\mathrm{j}(k))}\delta_x(k),
\end{multline}
where $\mathcal{A}(z_\mathrm{i}(k), z_\mathrm{j}(k))$ is the one-step-ahead map of the incremental state that depends on $z_\mathrm{i}(k)$, and $z_\mathrm{j}(k)$.\vspace{-3pt}
\begin{thm}\label{thm:contraction}
    Consider an LFR-based model augmentation structure in the form of \eqref{eqs:general_LFR}, satisfying both Assumption~\ref{assumption:baseline-lipschitz} and Condition~\ref{cond:ANN-lipschitz} and using the well-posed parametrization outlined in \eqref{eq:Dzw-final-parametrization}. Suppose there exists an $\mathcal{X}\succ 0$, and an $\bar{\alpha}\in\left(0, 1\right]$, such that
    \begin{equation}\label{eq:contracting_mdl_ineq}
        \bar{\alpha}^2\mathcal{X} - \mathcal{A}^\top(z_\mathrm{i}(k), z_\mathrm{j}(k))\mathcal{X}\mathcal{A}(z_\mathrm{i}(k), z_\mathrm{j}(k)) \succ 0, \quad \forall k>0.
    \end{equation}
    Then, \eqref{eqs:general_LFR} is contracting with some rate $\alpha<\bar{\alpha}$.
\end{thm}\vspace{-18pt}
\begin{pf}
    Note that, if \eqref{eq:contracting_mdl_ineq} holds, then there exists $\alpha<\bar{\alpha}$ such that
        \begin{equation}
        \alpha^2\mathcal{X} - \mathcal{A}^\top(z_\mathrm{i}(k), z_\mathrm{j}(k))\mathcal{X}\mathcal{A}(z_\mathrm{i}(k), z_\mathrm{j}(k)) \succeq 0.
    \end{equation}
    Then, based on \eqref{eq:delta_w_expression} we have
    \begin{equation}
        \alpha^2 V(\delta_{\hat{x}}(k)) \geq V(\delta_{\hat {x}}(k+1)),
    \end{equation}
    where $V(\delta_{\hat{x}}) = \delta_{\hat{x}}^\top\mathcal{X}\delta_{\hat{x}}$. Hence, the incremental dynamics form \eqref{eqs:incremental_dyn} is exponentially stable (in the Lyapunov sense), which implies that the augmented model \eqref{eqs:general_LFR} is contracting (see the proof of~\cite[Theorem 1]{revay_recurrent_2024}).\placeqed
\end{pf}\vspace{-15pt}
A sufficient condition for the contracting property can be obtained from Theorem~\ref{thm:contraction}, which is more suitable for deriving a stable-by-construction model parametrization.\vspace{-3pt}
\begin{cor}\label{corollary:contracting_param}
    Consider an LFR-based model augmentation structure with the assumptions as in Theorem~\ref{thm:contraction}. A sufficient condition for \eqref{eqs:general_LFR} to be contracting is
    \begin{equation}\label{eq:contraction_sufficient}
        \|\mathcal{A}(z_\mathrm{i}(k), z_\mathrm{j}(k))\|_2<\bar{\alpha}\leq 1, \quad \forall z_\mathrm{i}(k),z_\mathrm{j}(k)\in\mathbb{R}^{n_\mathrm{z}}.
    \end{equation}
\end{cor}\vspace{-18pt}
\begin{pf}
    Since $\mathcal{A}^\top(\cdot) \mathcal{A}(\cdot) = \|\mathcal{A}(\cdot)\|_2^2 <\bar{\alpha}^2$, \eqref{eq:contraction_sufficient} $\Rightarrow$ \eqref{eq:contracting_mdl_ineq} holds. E.g., select $\mathcal{X}=I$, then
    \begin{equation}
        \bar{\alpha}^2 I - \mathcal{A}^\top(z_\mathrm{i}(k), z_\mathrm{j}(k)) I \mathcal{A}(z_\mathrm{i}(k), z_\mathrm{j}(k)) \succ 0,
    \end{equation}
     for all $z_\mathrm{i}(k),z_\mathrm{j}(k)\in\mathbb{R}^{n_\mathrm{z}}$, s.t. $x_\mathrm{b}(k) \in \mathbb{X}$, $u(k) \in \mathbb{U}$.\placeqed
\end{pf}\vspace{-15pt}

\subsection{Direct parametrization of contracting LFR model augmentation structures}\vspace{-6pt}
We now provide a direct parametrization based on Corollary~\ref{corollary:contracting_param} for the matrices $A$, $B_\mathrm{w}$, and $C_\mathrm{z}$ to guarantee that the model structure \eqref{eqs:general_LFR} is contracting. Applying the triangle inequality and the sub-multiplicative property of the spectral norm to $\|\mathcal{A}(\cdot)\|_2$:
\begin{multline}\label{eq:mathcalA_factored}
    \|\mathcal{A}(\cdot)\|_2 \leq \|A\|_2 +\\ \|B_\mathrm{w}\|_2 \|\left(I_{n_w} - P^\mathrm{NL}(\cdot)D_\mathrm{zw}\right)^{-1}\|_2 \|P^\mathrm{NL}(\cdot)\|_2 \|C_\mathrm{z}\|_2.
\end{multline}
Moreover, based on the Neumann series identity~\cite{stewart_matrix_1998} and utilizing that $\|\Phi^\mathrm{NL}(k)\|_2\leq L$, the following bound can be expressed for the matrix-inverse in \eqref{eq:mathcalA_factored}:
\begin{equation}
    \|\left(I_{n_w} - P^\mathrm{NL}(\cdot)D_\mathrm{zw}\right)^{-1}\|_2 \leq \frac{1}{1 - L \|D_\mathrm{zw}\|_2}.
\end{equation}
Again, using the specified bound for $\|P^\mathrm{NL}(\cdot)\|_2$, condition~\eqref{eq:contraction_sufficient} can be modified as
\begin{equation}
    \|A\|_2 + \frac{L\|B_\mathrm{w}\|_2 \|C_\mathrm{z}\|_2}{1 - L\|D_\mathrm{zw}\|_2}<\bar{\alpha}.
\end{equation}
We introduce $X_\mathrm{A}$, $Y_\mathrm{A}$, \dots, $Z_\mathrm{C}$ free variables (with appropriate dimensions) to parametrize $\bar{A}$, $\bar{B}_w$, and $\bar{C}_z$ with the general Cayley transform, as shown in \eqref{eq:general_cayley}, hence $\|\bar{A}\|_2,\, \|\bar{B}_w\|_2,\,\|\bar{C}_z\|_2<1$. For notational convenience, the following expression is introduced:
{\setlength{\abovedisplayskip}{6pt}
 \setlength{\belowdisplayskip}{6pt}
\begin{equation}\label{eq:kappa}
    \kappa = \frac{L}{1-L\|D_\mathrm{zw}\|_2}.
\end{equation}
}\noindent
Furthermore, we assume that $D_\mathrm{zw}$ is parametrized as \eqref{eq:Dzw-final-parametrization}; thus, the value of $\kappa$ is known at each iteration. To guarantee that the condition for the contracting model parametrization is fulfilled, we introduce free scalar variables $\alpha$, $\beta$, $\gamma$, and construct $A$, $B_\mathrm{w}$, and $C_\mathrm{z}$ matrices, as
\begin{equation}
    A = \bar{\alpha}\, \sigma_A \bar{A}, \quad B_\mathrm{w} = \sigma_\mathrm{B} \bar{B}_\mathrm{w}, \quad C_\mathrm{z} = \sigma_\mathrm{C} \bar{C}_\mathrm{z},
\end{equation}
where $\sigma_A=\mathrm{sigmoid}(\alpha)$, guaranteeing that $\sigma_A\in(0,1)$, while~$\sigma_B$, and $\sigma_C$ are expressed as
\begin{align}
    \sigma_B &= \sqrt{\bar{\alpha}}\exp(\beta) \sqrt{\frac{1-\sigma_A}{\kappa}},\\
    \sigma_C &= \frac{\sqrt{\bar{\alpha}}}{\exp(\beta)} \sqrt{\frac{1-\sigma_A}{\kappa}} - \frac{\sqrt{\bar{\alpha}}\,\mathrm{sigmoid}(\gamma)}{\kappa \exp(\beta) \sqrt{\frac{1-\sigma_A}{\kappa}}}.
\end{align}
By following this parametrization the condition for contracting model structure is automatically fulfilled, since
\begin{equation}
    \|A\|_2 + \kappa \|B_\mathrm{w}\|_2 \|C_\mathrm{z}\|_2 < \bar{\alpha}\left(1 - \mathrm{sigmoid}(\gamma)\right)<\bar{\alpha}.
\end{equation}
This concludes Contribution 2. Fig.~\ref{fig:WP_contr_param} illustrates the relations between the free variables and LFR matrix components, as a compact overview of the proposed well-posed and contracting parametrizations.\vspace{-10pt}
%
%
\begin{figure}
    \centering
    \includegraphics[scale=0.85]{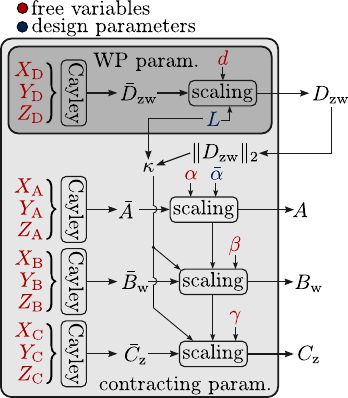}\vspace{-6pt}
    \caption{Illustrating the elements of the well-posed and contracting parametrizations.}
    \label{fig:WP_contr_param}
\end{figure}

%% file: sections/5_ident_algorithm.tex
\section{Identification algorithm}\label{sec:ident_algorithm}\vspace{-6pt}
Now that direct parametrizations for guaranteed well-posedness and stability properties of LFR-based model augmentation are presented in Sections \ref{sec:WellPosedness} and \ref{sec:stable_param}, respectively, the final step is to present an identification pipeline. First, various regularization options are introduced in this section, which can be added to the loss function \eqref{eq:cost-fun} for certain advantages. Then, finally, a highly efficient identification pipeline is adapted for the estimation of \eqref{eqs:general_LFR} that can handle the involved parameterization and also all considered regularization options.\vspace{-6pt}

\subsection{Standard $\ell_2$ and $\ell_1$ regularization}\vspace{-6pt}
To prevent overfitting, $\ell_2$ or $\ell_1$ regularization can be applied for the learning component parameters. Hence, the following terms can be added to cost function~\eqref{eq:cost-fun}:
{\setlength{\abovedisplayskip}{6pt}
 \setlength{\belowdisplayskip}{6pt}
\begin{equation}
    r_\mathrm{a}(\theta_\mathrm{a}) = \frac{\rho_\mathrm{a}^{\ell_2}}{2}\|\theta_\mathrm{a}\|_2^2 + \rho_\mathrm{a}^{\ell_1}\|\theta_\mathrm{a}\|_1,
\end{equation}
}\noindent
where $\rho_\mathrm{a}^{\ell_2}, \rho_\mathrm{a}^{\ell_1}\in\mathbb{R}_{\geq 0}$ correspond to the weights of the~$\ell_2$ and~$\ell_1$ penalties on the ANN parameters, respectively. Incorporating these terms, especially the $\ell_1$ penalty, promotes sparsity in the learning component parameters. Additionally, to avoid potential stability issues caused by tuning the initial states $\hat{x}_0$ excessively large, the following regularization term can also be applied:
{\setlength{\abovedisplayskip}{6pt}
 \setlength{\belowdisplayskip}{6pt}
\begin{equation}\label{eq:r_x0}
    r_\mathrm{x_0}(\hat{x}_0) = \frac{\rho_\mathrm{x_0}}{2}\|\hat{x}_0\|_2^2,
\end{equation}
}\noindent
where $\rho_\mathrm{x_0}\in\mathbb{R}_{\geq 0}$ denotes the $\ell_2$ regularization weight on~$\hat{x}_0$. Intuitively, \eqref{eq:r_x0} represents a prior that~$\hat{x}_0$ is small in absolute value, which is justified given a proper experiment design. Incorporating $r_\mathrm{x_0}$ into the cost function can improve the conditioning of the optimization problem by adding a positive definite contribution to the Hessian matrix of the objective function w.r.t. the initial state.\vspace{-6pt}

\subsection{Regularization of the baseline parameters}\label{sec:phys-parm-regul}\vspace{-6pt}
The LFR-based model augmentation structure is inherently overparametrized, i.e., there exist multiple parameter triplets $(\theta_\mathrm{b}^\ast, \theta_\mathrm{a}^\ast, \theta_\mathrm{LFR}^\ast)$ that minimize \eqref{eq:cost-fun}. Non-uniqueness of $\theta_\mathrm{a}^\ast$, and~$\theta_\mathrm{LFR}^\ast$ does not affect the interpretability of the resulting model; however, due to the non-uniqueness of $\theta^\ast_\mathrm{b}$, the baseline parameters may converge to unrealistic values, ultimately undermining the physical interpretability of the final model. Additionally, physics-based models are typically capable of extrapolating accurately beyond the range of measurement data, while ANNs usually showcase poor extrapolation capabilities outside of the training data set. With physically unrealistic baseline parameters, the augmented model might also lose this advantageous extrapolation capability. To address this issue, a regularization cost term can be added to \eqref{eq:cost-fun} based on \cite{bolderman_physics-guided_2024}, as
\begin{equation}\label{eq:phys_regul}
    r_\mathrm{b}(\theta_\mathrm{b}) = \dfrac{\rho_\mathrm{b}}{2} \left\|\Lambda_\mathrm{b} \left(\theta_\text{b} - \theta_\text{b}^0\right)\right\|_2^2,
\end{equation}
where $\Lambda_\mathrm{b} = \mathrm{diag}\left(\theta_\mathrm{b}^0\right)^{-1}$, and $\rho_\mathrm{b} \in\mathbb{R}_{\geq 0}$ is a trade-off parameter. 
Applying \eqref{eq:phys_regul} ensures that the baseline parameters remain close to their nominal values, while deviations are allowed when it results in a better data fit.\vspace{-6pt}

\subsection{Discovery of augmentation structures}\label{sec:augm_struct_discovery}\vspace{-6pt}
To automatically detect the most efficient model augmentation structure for the given identification problem, an $\ell_1$ regularization term can be imposed on the elements of the~$W_{\theta_\mathrm{LFR}}$ matrix, as
\begin{equation}\label{eq:W_LFR_L1_reg}
    r_\mathrm{LFR}(\theta_\mathrm{LFR}) = \rho_\mathrm{LFR} \|\Lambda_\mathrm{LFR} \mathrm{vec}(W_{\theta_\mathrm{LFR}})\|_1,
\end{equation}
where $\rho_\mathrm{LFR}\in\mathbb{R}_{\geq 0}$ is a collective trade-off parameter, while $\Lambda_\mathrm{LFR}=\mathrm{diag}(\lambda_\mathrm{LFR}^1,\lambda_\mathrm{LFR}^2,\dots,\lambda_\mathrm{LFR}^{n_\mathrm{W}})$ collects the regularization weights associated with each entry of $W_{\theta_\mathrm{LFR}}$. The coefficients $\lambda_\mathrm{LFR}^i\in\mathbb{R}_{\geq0}$ control the magnitude of the $\ell_1$ penalty term regarding each entry in $W_{\theta_\mathrm{LFR}}$ and determine the trade-off between model complexity (in terms of augmentation structure) and data fit. A simple choice is to set all weights uniformly. A more systematic way is to employ the iterative reweighting scheme according to \cite{candes_enhancing_2008}, summarized as follows:\vspace{-2mm}
\begin{enumerate}
    \item Set $\tau{=}0$ and initialize $\lambda_\mathrm{LFR}^i(\tau){=}0$ for $i{=}1,\dots,n_\mathrm{W}$.
    \item Solve the optimization problem \eqref{eq:optim-problem} with \eqref{eq:W_LFR_L1_reg} using $\Lambda_\mathrm{LFR}(\tau)=\mathrm{diag}(\lambda_\mathrm{LFR}^1(\tau),\dots,\lambda_\mathrm{LFR}^{n_\mathrm{W}}(\tau))$, yielding the estimate $\theta(\tau)=\mathrm{vec}(\theta_\mathrm{b}(\tau),\theta_\mathrm{a}(\tau),\theta_\mathrm{LFR}(\tau))$.
    \item Update the weights for $i=1,\dots, n_\mathrm{W}$ according to
    \begin{equation}
        \lambda_\mathrm{LFR}^i(\tau+1) = \frac{1}{\vert [\mathrm{vec}({W_{\theta_\mathrm{LFR}(\tau)})]_i\vert+\varepsilon}},
    \end{equation}
    with a constant $0<\varepsilon\ll 1$ to avoid singularities. 
    \item Terminate if the weights converge or when a prescribed maximum number of iterations is reached; otherwise, increment $\tau$ and return to Step~2.
    \item Fix all elements in $W_{\theta_\mathrm{LFR}}$ to zero for which the iteration converged to a value that is smaller in absolute value than $\varepsilon$. Then, re-estimate model \eqref{eqs:general_LFR} with the sparsified $W_{\theta_\mathrm{LFR}}$ matrix and $\rho_\mathrm{LFR}=0$. 
\end{enumerate}\vspace{-2mm}
Ideally, an $\ell_0$ penalty would be applied to promote sparsity of the $W_{\theta_\mathrm{LFR}}$ matrix, but it is often computationally challenging to exactly solve the $\ell_0$ regularization task. Instead, typically $\ell_1$ regularization is applied as a substitute. On the other hand, uniform $\ell_1$ regularization penalizes the magnitude of each parameter, introducing bias in sparsification. The iterative reweighting strategy aims to resolve this issue. 
By appropriate hyperparameter selection, it recovers the underlying sparse components of $W_{\theta_\mathrm{LFR}}$, i.e., detects the redundant components in the LFR matrix, and facilitates automatic augmentation structure discovery. 
Naturally, the collective trade-off parameter $\rho_\mathrm{LFR}$ influences the overall effect of the sparsity penalty on the estimation. Selecting an appropriate value might not be intuitive; hence, we propose the following rule of thumb:
{\setlength{\abovedisplayskip}{6pt}
 \setlength{\belowdisplayskip}{6pt}
\begin{equation}\label{eq:rho_LFR_rule_of_thumb}
    \rho_\mathrm{LFR} = \varepsilon V_{\mathcal{D}_N}^\mathrm{base},
\end{equation}
}\noindent
where $V_{\mathcal{D}_N}^\mathrm{base}$ is the value of the cost function \eqref{eq:cost-fun} evaluated using only the baseline model. With this $\rho_\mathrm{LFR}$ selection, the regularization term $r_\mathrm{LFR}$ remains in a similar order of magnitude as $V_{\mathcal{D}_N}$, and the optimization can focus on improving the data-fit while simultaneously penalizing the redundant elements of $W_{\theta_\mathrm{LFR}}$.\vspace{-6pt}

\subsection{Model order selection}\vspace{-6pt}
The dimensions of the latent variables $z_\mathrm{a}$, and $w_\mathrm{a}$ should be selected carefully. By increasing $n_{\mathrm{z_a}}$ and $n_{\mathrm{w_a}}$, the capability of the model to capture nonlinear effects can be enhanced. Similarly, determining the dimension of the augmented state $x_\mathrm{a}$ is not intuitive. Based on physical knowledge, a good baseline value can often be provided for $n_{\mathrm{z_a}}$, $n_\mathrm{w_a}$, and $n_\mathrm{x_a}$. 
Alternatively, by adapting the methodology discussed in \cite{bemporad_l-bfgs-b_2025}, a \emph{group-lasso} penalty can be applied to select the dimensions of $z_\mathrm{a}$, $w_\mathrm{a}$, and $x_\mathrm{a}$, as
{\setlength{\abovedisplayskip}{6pt}
 \setlength{\belowdisplayskip}{6pt}
\begin{multline}\label{eq:group-lasso}
    r_\mathrm{g}(\theta, \hat{x}_0) = \rho_\mathrm{z_a} \sum_{i=1}^{n_{\mathrm{z_a}}} \left\|\theta_{\mathrm{z_a},i}\right\|_2 + \rho_\mathrm{w_a} \sum_{i=1}^{n_{\mathrm{w_a}}} \left\|\theta_{\mathrm{w_a},i}\right\|_2 + \\
    \rho_\mathrm{x_a} \sum_{i=1}^{n_{\mathrm{x_a}}} \left\|\theta_{\mathrm{x_a},i}\right\|_2,
\end{multline}
}\noindent
where $\rho_\mathrm{z_a}, \rho_\mathrm{w_a}, \rho_\mathrm{x_a} \geq 0$ are regularization weights, and $\theta_{\mathrm{z_a},i},\, \theta_{\mathrm{w_a},i}, \,\theta_{\mathrm{x_a},i}$ all represent groups of parameters that correspond to the specific dimension of $z_\mathrm{a},\, w_\mathrm{a},\, x_\mathrm{a}$, respectively. For example, when $\rho_\mathrm{w_a}>0$, $\theta_{\mathrm{w_a},i}$ collects all parameters associated with the $i$\textsuperscript{th} dimension of $w_\mathrm{a}$: the $i$\textsuperscript{th} row of and $i$\textsuperscript{th} element of the weight and bias of the last layer in the learning component, and the $i$\textsuperscript{th} columns of $D_\mathrm{zw}^\mathrm{ba}$, $D_\mathrm{zw}^\mathrm{aa}$, $B_\mathrm{w}^\mathrm{ba}$, $B_\mathrm{w}^\mathrm{aa}$, $D_\mathrm{yw}^\mathrm{ba}$, and $D_\mathrm{yw}^\mathrm{aa}$. The basic idea of the group-lasso regularization in \eqref{eq:group-lasso} is to see whether specific dimensions of $z_\mathrm{a}$, $w_\mathrm{a}$ or $x_\mathrm{a}$ can be eliminated by zeroing out the corresponding parameters, therefore automating the process of model order selection. Although~\eqref{eq:group-lasso} is formulated for all three mentioned variables, in practice, it is advised to apply group-lasso regularization for only one variable at a time and repeat the process until finding appropriate dimensions for all $z_\mathrm{a}$, $w_\mathrm{a}$, and $x_\mathrm{a}$. Practical advice about applying the proposed approach will be given in Section~\ref{sec:group-lasso-example}.\vspace{-6pt}

\subsection{Data normalization and parameter initialization}\vspace{-6pt}
IO normalization is required for efficient training of artificial neural networks,  mainly to avoid exploding and vanishing gradients \cite{lecun_efficient_1998}. For black-box ANN-SS models, the data normalization process is usually implemented before the optimization, and back-scaling is applied after each model evaluation to compare the results with the true system outputs~\cite{beintema_deep_2023}. However, this is not applicable for grey-box approaches, as the baseline model uses the physically true, unscaled data. Hence, we apply similar transformations as discussed in~\cite{schoukens_initialization_2020} to scale and back-scale the baseline component. For more details, refer to~\cite{hoekstra_learning-based_2025}. Besides IO normalization, a reliable parameter initialization scheme is also important to potentially improve accuracy and decrease convergence time, since fully random initialization might cause stability losses at the beginning of the optimization. To address these challenges, a similar initialization method is applied as in \cite{hoekstra_learning-based_2025}, which guarantees that the augmentation structure behaves like the FP model at the beginning of training, as it is assumed to provide reasonably accurate results.\vspace{-6pt}

\subsection{Optimization algorithm}\vspace{-6pt}
To handle $\ell_1$-regularization and non-smooth cost terms, we adapt the methodology described in \cite{bemporad_l-bfgs-b_2025} by splitting the model parameters as $\theta=\theta_+ - \theta_-$, such that $\theta_+, \theta_- \geq 0$. Moreover, we apply the L-BFGS-B optimization algorithm \cite{byrd_limited_1995} to solve \eqref{eq:optim-problem}. For warm-starting the optimization, a fixed number of Adam \cite{kingma_adam_2015} gradient-descent steps are also applied to prevent converging to local minima. 
For computing the gradients w.r.t. $\theta$ and $\hat{x}_0$, a highly efficient automatic differentiation tool, namely the {\sc JAX} library~\cite{bradbury_jax_2018}, is used. A key challenge is the backpropagation through the \emph{fixed-point iterations} (FPI) performed when evaluating~\eqref{eqs:general_LFR}. The FPI in Proposition~\ref{prop:FPI} is implemented to stop after $z_n$ changes less than a user-specified tolerance value $\varepsilon_\mathrm{FPI}\ll 1$ or the number of iterations exceeds a pre-defined limit $n_\mathrm{max}$.  
Then, we utilize the implicit function theorem and automatic implicit differentiation~\cite{blondel_efficient_2022} to directly calculate gradients of the fixed-point $z_\ast(k)$ w.r.t. model parameters. The proposed regularization options and the adaptation of the efficient identification pipeline from \cite{bemporad_efficient_2025} into the LFR-based model augmentation setting conclude Contribution~3 of the paper.\vspace{-6pt}

\subsection{Initial states for testing}\label{sec:init_states_testing}\vspace{-6pt}
For testing the trained models on a new data set $\tilde{\mathcal{D}}_{N_\mathrm{test}} = \{(\tilde{y}(i), \tilde{u}(i))\}_{i=0}^{N_\mathrm{test}-1}$, an initial state estimate is necessary to start the simulation. Throughout this paper, we solve the following optimization problem to estimate the initial state of the test data:
{\setlength{\abovedisplayskip}{6pt}
 \setlength{\belowdisplayskip}{6pt}
\begin{equation}\label{eq:optim-problem-x0}
    \min_{\tilde{x}_0} \quad \dfrac{1}{N_\mathrm{i}}\sum_{k=0}^{N_\mathrm{i}} \left\|\tilde{y}(k) - \hat{y}(k)\right\|_2^2,
\end{equation}
}\noindent
where $N_\mathrm{i} \leq N_\mathrm{test}$ is the length of the state initialization window, while $\hat{y}(k)$ is the model output obtained by simulating the trained model structure \eqref{eqs:general_LFR} from $\tilde{x}_0$. We solve \eqref{eq:optim-problem-x0} with the L-BFGS algorithm, and an initial guess for the optimization is provided by running an \emph{Extended Kalman Filter}~(EKF) and \emph{Rauch-Tung-Striebel}~(RTS) smoothing backward in time~$N_\mathrm{e}$ times. For $N_\mathrm{i} = N_\mathrm{test}$, the methodology is the same as in~\cite{bemporad_l-bfgs-b_2025}.\vspace{-6pt}

%% file: sections/6_num_examples.tex
\section{Examples}\label{sec:num_examples}\vspace{-6pt}
We apply the discussed model augmentation approach on various examples to demonstrate the efficiency of the proposed methodology. 
Code implementation and experimental data are available on GitHub\footnote{\scriptsize \url{https://github.com/AIMotionLab-SZTAKI/StableLFRaugmentation}}.\vspace{-6pt}

\subsection{F1Tenth simulation example}\vspace{-6pt}
We evaluate the proposed approaches by identifying the dynamics of a small-scale electric vehicle (the F1Tenth car \cite{agnihotri_teaching_2020}). To analyze the methods through various noise levels, data is generated with a high-fidelity digital twin of the car. Then, i.i.d. Gaussian white noise is added to reach specific \emph{signal-to-noise ratio} (SNR) values for the training data, while test data is kept noise-free. 
Data were collected from multiple experiments using circular and lemniscate trajectories at various speeds, yielding $N=15985$ samples, split evenly between training and testing (see \cite{gyorok_orthogonal_2025} for details).
%

The digital twin is treated as a black box, so the true dynamics are unknown. Instead, we use a simplified mechanical model based on the single-track representation, which approximates the velocity dynamics.\footnote{\scriptsize Only the velocity dynamics $(v_\xi, v_\eta, \omega)$ are considered for augmentation. After identification, the position and orientation values can be computed using the integrator dynamics of the single-track representation. See \cite{gyorok_orthogonal_2025} for more details.} The inputs are the steering angle $\delta(k)$ and the PWM percentage $d(k)$ of the electric motor, while the states are 
the longitudinal and lateral velocities, $v_\xi(k)$ and $v_\eta(k)$, and the yaw rate $\omega(k)$. The baseline model is expressed as
\begin{subequations}
{\setlength{\abovedisplayskip}{6pt}
 \setlength{\belowdisplayskip}{0pt}
\begin{multline}
    v_\xi(k+1) = v_\xi(k) + \frac{T_\mathrm{s}}{m}(F_\xi + F_\xi \cos\delta(k) -\\ F_\mathrm{f,\xi}\sin\delta(k)-mv_\eta(k)\omega(k)),
\end{multline}}
{\setlength{\abovedisplayskip}{0pt}
 \setlength{\belowdisplayskip}{0pt}
\begin{multline}
    v_\eta(k+1) = v_\eta(k) + \frac{T_\mathrm{s}}{m}(F_\mathrm{r,\eta} + F_\xi\sin\delta(k)-\\F_\mathrm{f,\eta}\cos\delta(k)-mv_\xi(k)\omega(k)),
\end{multline}}
{\setlength{\abovedisplayskip}{0pt}
 \setlength{\belowdisplayskip}{6pt}
\begin{multline}
    \omega(k+1) = \omega(k) + \frac{T_\mathrm{s}}{J_\mathrm{z}}(F_\mathrm{f,\eta} l_\mathrm{f}\cos\delta(k)+\\F_\xi l_\mathrm{f}\sin\delta(k)-F_\xi l_\mathrm{r}),
\end{multline}
}\noindent
\end{subequations}
where $T_\mathrm{s}=0.025~\mathrm{s}$ is the sampling time, $F_\xi$, $F_\mathrm{f,\eta}$, $F_\mathrm{r,\eta}$ denote the tire forces, which are computed based on a semi-empirical tire model (see~\cite{gyorok_orthogonal_2025} for more details). These force terms depend on the state and input values that are omitted for simplicity. Furthermore, $m$ is the mass of the vehicle, $J_\mathrm{z}$ is the inertia along the vertical axis, $l_\mathrm{r}$ and $l_\mathrm{f}$ are the distances of the rear and front axis from the center of gravity. Additionally, the tire model contains 5 more parameters; hence, altogether the baseline model has 9 parameters that are co-estimated with the learning component, i.e., $\theta_\mathrm{b}\in\mathbb{R}^9$. All states are directly measured in the high-fidelity simulation environment as it would be for the real F1Tenth car, i.e., $\hat{y}_\mathrm{b}(k)=x_\mathrm{b}(k)$ output equation is applied for the baseline model. 

Key learning settings are summarized in Table~\ref{tab:F1Tenth_hyperparams}. 
Following \cite{hoekstra_learning-based_2025}, two scenarios are considered: static augmentation ($n_{\mathrm{x_a}}=0$), and dynamic augmentation ($n_{\mathrm{x_a}}=2$). Since the baseline states are measured, $\hat{x}_0$ is not included in the optimization variables for static augmentation, only in the dynamic case. Consequently, initial state estimation for testing (see Section~\ref{sec:init_states_testing}) is only required for dynamic augmentation. For this purpose, we apply $N_\mathrm{i}=12$ based on \cite{szecsi_deep_2024}. Latent dimensions are set to $n_{\mathrm{z_a}}=4$ , $n_{\mathrm{w_a}}=3+n_{\mathrm{x_a}}$. 
In addition to the proposed well-posedness (denoted as WP in the remainder) and contracting parametrization, we evaluate two strategies from~\cite{hoekstra_learning-based_2025}: 
when $D_{zw} \equiv 0$ and when $D_{zw}^\mathrm{ab}$ is tuned, but the other elements of $D_{zw}$ are set to zero. All models are trained with the same hyperparameters from 10 initializations, and the best-performing model on the training set is selected for each parametrization.
\begin{table}
    \centering
    \scriptsize
    \renewcommand{\arraystretch}{1.1}
    \caption{Learning-related options for the F1Tenth example.}
    \label{tab:F1Tenth_hyperparams}
    \begin{tabular}{lc}
    \hline
        Parameter name & Value\\
    \hline
        ANN structure & 1 hidden layer, 128 nodes\\
        Activation function & \emph{hyperbolic tangent} (tanh)\\
        Adam epochs & 4000\\
        L-BFGS-B epochs & 5000\\
        Lipschitz bound $L$ & 2.5\\
        Lipschitz regul. coeff. $\rho_\mathrm{L}$ & 0.1\\
        FPI max. iterations $n_\mathrm{max}$ & 10\\
        FPI tolerance $\varepsilon_\mathrm{FPI}$ & $10^{-5}$\\
        Baseline param. regul. coeff. $\rho_\mathrm{b}$ & 0.005\\
    \hline
    \end{tabular}\vspace{-1mm}
\end{table}

Test results are summarized in Table~\ref{tab:f1tenth_results}, using the \emph{normalized root-mean-squared error} (NRMSE) metric. The baseline model with $\theta_\mathrm{b}^0$ achieves 37.91\% test NRMS error; compared to which all LFR-based structures result in significant improvements. The simplest approach for guaranteeing well-posedness, namely fixing $D_{zw}=0$, is consistently outperformed by the other two methods. This is especially visible for the dynamic case. This simplest strategy can be interpreted as a hybrid additive-multiplicative augmentation structure; thus, it is visible that the more general LFR-based approaches with a non-zero $D_\mathrm{zw}$ matrix outperform the conventional augmentation structures. Tuning only $D_\mathrm{zw}^\mathrm{ab}$ resulted in similar model accuracies compared to the WP parametrization; however, the latter applies significantly fewer constraints on the $D_\mathrm{zw}$ matrix, hence it can be tuned more freely. This can be the reason why the WP parametrization method performs better during higher noise levels. The more rigid parametrization with triangular $D_\mathrm{zw}$ makes it harder to separate the true dynamics and the noise. The contracting model parametrization does not achieve as high test accuracy as the other methods, due to the stable-by-design parametrization being restrictive, i.e., not all stable models can be characterized by this method. In return, the contracting LFR-based structure guarantees stability by design. Nonetheless, a high performance gain is still achieved compared to the baseline model. Moreover, these test errors are comparable to results with the standard static additive model augmentation structure\footnote{\scriptsize That paper applied a slightly different objective function, but the similarity of the identification task allows for comparison.} in \cite{gyorok_orthogonal_2025}, which achieved 3.84\% NRMSE on the same test set with 30 dB SNR levels during training. Therefore, a slight reduction in model accuracy can be considered an acceptable trade-off for the added benefits of the contracting property. A similar trade-off has been observed for contracting model parametrizations in the black-box system identification setup~\cite{revay_recurrent_2024}.
\begin{table}
    \centering
    \scriptsize
    \renewcommand{\arraystretch}{1.1}
    \caption{Test results on the F1Tenth identification example with different noise levels regarding the training data set.}
    \label{tab:f1tenth_results}
    \begin{tabular}{lccccc}
    \hline
        Type & Parametr. & None & 40 dB & 30 dB & 20 dB\\
    \hline
        \multirow{4}{*}{Static} & $D_{zw}\equiv 0$ & {\bf 1.96\%} & 2.02\% & 2.48\% & 4.67\%\\
         & $D_{zw}^\mathrm{ab}$ tuned & 1.97\% & {\bf 1.96\%} & {\bf 2.27\%} & 4.35\%\\
        & WP & {\bf 1.96\%} & 2.06\% & 2.57\% & {\bf 4.07\%}\\
         & contract. & 4.61\% & 4.44\% & 4.12\% & 4.27\%\\
    \hline
        \multirow{4}{*}{Dyn.} & $D_{zw}\equiv 0$ & 1.56\% & 1.29\% & 1.68\% & 3.03\%\\
        & $D_{zw}^\mathrm{ab}$ tuned & {\bf 1.09\%} & {\bf 1.10\%} & 1.47\% & 2.91\%\\
        & WP & 1.16\% & 1.25\% & {\bf 1.43\%} & {\bf 2.75\%}\\
        & contract. & 2.17\% & 2.19\% & 2.24\% & 3.52\%\\
    \hline
    \end{tabular}\vspace{-8pt}
\end{table}

The training times\footnote{\scriptsize Using Python~3.10 on a system with an Intel Core i7-13700HX processor and 16~GB of RAM.} for each model parametrization are reported in Fig.~\ref{fig:training_times}, based on a Monte Carlo study with 10 independent runs (all with different initialization). The results refer to the dynamic structures trained under the noiseless setting, as excluding the augmented states or training at different SNR levels did not significantly affect the training times. The WP and contracting parametrizations exhibit increased computational cost due to the required fixed-point iterations and the numerically expensive Cayley transform. Nevertheless, a state-of-the-art DT ANN-SS model on a similar identification problem \cite{szecsi_deep_2024} required approximately 15 minutes of training, which is on par with the WP-LFR model. Moreover, the simplified $D_\mathrm{zw}$ structures yield noticeably reduced training times compared to the black-box parametrization, highlighting the benefits of the utilized identification pipeline. When compared to continuous-time ANN-SS models, which required roughly two hours of training in \cite{szecsi_deep_2024}, even the contracting LFR-based parametrization demonstrates a substantial reduction in training time.
\begin{figure}
    \centering
    \includegraphics{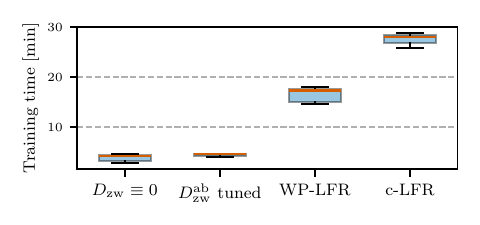}
    \vspace{-12pt}
    \caption{Training times of the dynamic LFR-based model augmentation structures with different parametrizations. The results are averaged over 10 Monte Carlo runs.}\vspace{-1mm}
    \label{fig:training_times}
\end{figure}
%
%

To highlight the main advantage of the contracting parametrization, we have simulated the WP and contracting LFR-based models from various initial conditions, some even outside of the state space covered by the training set, on an arbitrary test trajectory. The simulated $v_\xi$ values are shown in Fig.~\ref{fig:F1Tenth_stab}. Remarkably, the contracting model responses all converge to the nominal trajectory within 2 seconds of simulation time. For certain initial points, a similar trend can be seen for the WP structure. However, on many occasions (where $\hat{x}_0$ lay outside the training regime), the WP structure has shown slower convergence to the nominal model response compared to the contracting model; moreover, static and dynamic stability losses are both visible.\vspace{-6pt}
\begin{figure}
    \centering
    \includegraphics{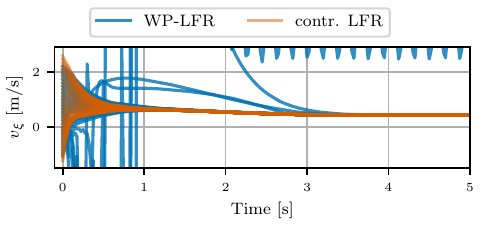}\vspace{-12pt}
    \caption{WP and contracting LFR-based structures simulated on an arbitrary test trajectory from various initial conditions.}\vspace{-1mm}
    \label{fig:F1Tenth_stab}
\end{figure}

\subsection{Cascaded Tanks benchmark example}\vspace{-6pt}
The \textit{Cascaded Tanks with Overflow}~(CCT) benchmark~\cite{schoukens_three_2017} is a two-tank system where a pump feeds the upper tank, and water flows into the lower tank through a small opening. The input signal controls the pump, and the measured output is the water level of the lower tank. Both tanks can overflow, introducing saturation. Moreover, when the upper tank overflows, excess water may spill into the lower tank, leading to further complexity. The benchmark contains two datasets, one for estimation and one for testing, with measurements of 1024 samples each at a sampling time of $T_\mathrm{s}=4~\mathrm{s}$.

We use a first-principles model based on the Bernoulli effect that neglects the overflow; hence, the baseline model dynamics can be expressed as
\begin{subequations}\label{eq:CCT_baseline}
{\setlength{\abovedisplayskip}{6pt}
 \setlength{\belowdisplayskip}{6pt}
\begin{align}
    \hat{x}_1(k+1) &= \hat{x}_1(k) + T_\mathrm{s} \left(-k_1 \sqrt{\hat{x}_1(k)} + k_2 u(k)\right),\\
    \hat{x}_2(k+1) &= \hat{x}_2(k) + T_\mathrm{s}\left(k_3 \sqrt{\hat{x}_1(k)} - k_4 \sqrt{\hat{x}_2(k)}\right),
\end{align}
}\noindent
\end{subequations}
where $u(k){\in}\mathbb{R}$ is the input, $\hat{x}_1(k),\hat{x}_2(k){\in}\mathbb{R}$ are the baseline states. The output equation is given as $\hat{y}_\mathrm{b}(k) = \hat{x}_2(k)$. The baseline model parameters $\theta_\mathrm{b} = [k_1\, k_2\, k_3\, k_4]^\top$ are co-estimated with other model parameters and initialized as $k_1^0=k_2^0=k_3^0=k_4^0=0.05$.

No additional states are introduced ($n_{\mathrm{x_a}}=0$) since the baseline states provide sufficient information. By trial and error, we select $n_{\mathrm{z_a}}=3$ and $n_{\mathrm{w_a}}=2$, with the learning component being a single-hidden-layer ANN with 16 neurons. The baseline model is estimated to be 1-Lipschitz, but a $L=2$ bound is applied for robustness. 
For the FPI, $n_\mathrm{max}=8$ and $\varepsilon_\mathrm{FPI}=10^{-6}$ are chosen. Regularization is applied with $\rho_\mathrm{b}=10^{-2}$ and $\rho_\mathrm{a}^{\ell_2}=10^{-5}$ to avoid unrealistic baseline parameters and overfitting. The contracting parametrization has shown slower convergence on this example; thus, hyperparameters have been separately tuned for each structure (see Table~\ref{tab:CT_hyperparams}).
\begin{table}
    \centering
    \scriptsize
    \renewcommand{\arraystretch}{1.1}
    \caption{Hyperparameters for the Cascaded Tanks example.}
    \label{tab:CT_hyperparams}
    \begin{tabular}{lcc}
    \hline
        Hyperparameter & WP param. & contract. param.\\
    \hline
        ANN activation fun. & swish\tablefootnote{\scriptsize By "swish", we refer to $\mathrm{swish}(x)=x \cdot\mathrm{sigmoid}(x)$.} & tanh\\
        Lipschitz regul. coeff. $\rho_\mathrm{L}$ & $10^{-3}$ & 1\\
        Adam epochs & 1000 & 2000\\
        L-BFGS-B epochs & 100 & 100\\
    \hline
    \end{tabular}\vspace{-1mm}
\end{table}

The initial states $\hat{x}_0$ are co-estimated with model parameters during training and initialized as $\hat{x}_1(0) = \hat{x}_2(0) = y(0)$, following~\cite{donati_combining_2025}. For testing, an initialization window of $N_\mathrm{i}=50$ is used to estimate the initial states, and performance is reported using the RMSE metric after excluding the first $N\mathrm{i}$ samples.\footnote{\scriptsize See benchmark rules: \url{https://www.nonlinearbenchmark.org/}} Results are shown in Table~\ref{tab:CCT_results}, including prior methods from the literature. Only two approaches ~\cite{dehkordi_normalization_2026,rogers_grey_2017} outperform our method, both relying on a realistic physics-based overflow model\footnote{\scriptsize The results in \cite{dehkordi_normalization_2026} rely on a novel normalization scheme developed for gray-box models, which could in principle be incorporated into our approach to further improve its performance. 
}, i.e., a far more accurate baseline model than \eqref{eq:CCT_baseline}. Notably, the WP-LFR model structure performs better than all black-box methods, and also~\cite{donati_combining_2025}, which uses an additive augmentation structure with a novel regularization scheme. This highlights an important advantage of the LFR-based structure, that it 
can find the best-performing model augmentation formulation during model learning. Simulation results on the test data are shown in Fig.~\ref{fig:CCT_output}. The contracting parametrization is less accurate due to the trade-off between guaranteed stability and model accuracy. Nonetheless, it remains close to conventional black-box methods such as LSTM.

This benchmark employs less training data and involves dynamics that can be captured with lower-complexity models compared to the F1Tenth example, making the efficiency of the adapted identification pipeline particularly evident. Training the WP parametrization takes about 12 seconds, while the contracting LFR-based structure converges in roughly 28 seconds. These training times are remarkably short for models with stability-enforcing parametrizations, highlighting the computational efficiency of the proposed approach.\vspace{-6pt}
\begin{table}
    \centering
    \scriptsize
    \renewcommand{\arraystretch}{1.1}
    \caption{Test RMSE on the CCT benchmark.}\vspace{2pt}
    \label{tab:CCT_results}
    \begin{tabular}{lc}
    \hline
        Method & RMSE\\
    \hline
        LSTM \cite{champneys_baseline_2024} & 0.45\\
        OLSTM \cite{champneys_baseline_2024} & 0.43\\
        PNARX \cite{champneys_baseline_2024} & 0.42\\
        dynoNet \cite{forgione_dynonet_2021} & 0.42\\
        DT SUBNET \cite{beintema_deep_2023} & 0.37\\
        CT SUBNET \cite{beintema_continuous-time_2022}  & 0.31\\
        Off-white model (physics-based regul.) \cite{donati_combining_2025} & 0.26\\
        Gray-box ANN-SS (iterative, trainable norm.) \cite{dehkordi_normalization_2026} & 0.22\\
        Grey-Box (physical overflow model)
        \cite{rogers_grey_2017} & 0.18\\
    \hline
        Applied initial baseline model & 2.62\\
        Contracting LFR model augmentation & 0.64\\
        WP LFR model augmentation & 0.25\\
    \hline
    \end{tabular}\vspace{-6pt}
\end{table}
\begin{figure}
    \centering
    \includegraphics{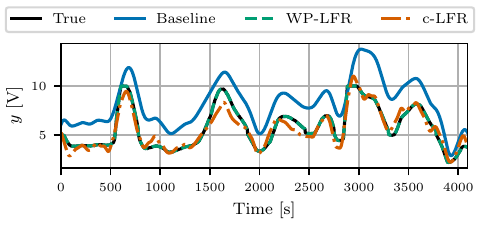}
    \vspace{-10pt}
    \caption{Simulation results on the CCT benchmark test data.}\vspace{-1mm}
    \label{fig:CCT_output}
\end{figure}

%

\subsection{Model augmentation structure discovery}\label{sec:augm_struct_discovery_example}\vspace{-6pt}
With the following example, we demonstrate the reweighted $\ell_1$ regularization scheme outlined in Section~\ref{sec:augm_struct_discovery} to automatically discover the underlying model augmentation structure. 
We consider the following DT \emph{mass-spring-damper} (MSD) data-generating system:
\begin{subequations}\label{eq:MSD}
{\setlength{\abovedisplayskip}{6pt}
 \setlength{\belowdisplayskip}{6pt}
\begin{align}
    x_1(k+1) &= x_1(k) + T_\mathrm{s} x_2(k),\\
    x_2(k+1) &= x_2(k) + \frac{T_\mathrm{s}}{m} (-k_\mathrm{s} x_1(k) + c_\mathrm{d} x_2(k) -\\ &\quad\quad F_\mathrm{fric}(x_2(k)) + u(k)),\nonumber
\end{align}
}\noindent
\end{subequations}
where $x_1(k)$ represents the measured position, i.e., $y(k) = x_1(k)$, and $x_2(k)$ is the velocity, $u(k)$ is an auxiliary force term, representing the input, $m$ is the mass, $k_\mathrm{s}$ is the spring constant, $c_\mathrm{d}$ is the damping coefficient and $T_\mathrm{s}$ is the sampling time. Furthermore, $F_\mathrm{fric}$ is a nonlinear friction term consisting of Stribeck-like and viscous components. The data-generating system uses $m=1$~kg, $k_\mathrm{s}=10^3~\mathrm{N/m}$, $c_\mathrm{d}=90~\mathrm{Ns/m}$, and $T_\mathrm{s}=0.01$~s. The baseline model is constructed by removing $F_\mathrm{fric}$ from \eqref{eq:MSD}, and all physical parameters are assumed to be known; hence, $\theta_\mathrm{b}$ is not co-estimated with the other parameters. Training data consists of $10^4$ samples generated with $u_k\sim\mathcal{U}(-100, 100)$, and $10^3$ additional samples are used for testing with the same input signal distribution. The WP-LFR structure is trained using the iterative reweighting scheme from Section~\ref{sec:augm_struct_discovery}. The learning component is an ANN with 1 hidden layer and 8 neurons, using the tanh activation. No augmented states are used ($n_{\mathrm{x_a}}=0$), while $n_\mathrm{z_a}$ and $n_\mathrm{w_a}$ are both set to 1. Fixed-point iterations use $n_\mathrm{max}=10$ and $\varepsilon_\mathrm{FPI}=10^{-5}$. The baseline Lipschitz constant is $10.05$, and regularization is applied with $\rho_\mathrm{L}=0.1$. The coefficient $\rho_\mathrm{LFR}$ is set via \eqref{eq:rho_LFR_rule_of_thumb}. Each reweighting step applies 500 Adam epochs followed by 1000 L-BFGS-B epochs. The scheme converged rapidly in just 4 iterations, sparsifying 70\% of the $W_{\theta_\mathrm{LFR}}$ matrix. 
The reweighting algorithm is applied with $\varepsilon=10^{-3}$, and remarkably, it yields a sparse, interpretable augmentation structure. For the output function, only the baseline output map $h_{\theta_\mathrm{b}}^\mathrm{FP}$ is utilized; all other components are zeroed out. For the state transition, 
a simple additive structure emerged, as
{\setlength{\abovedisplayskip}{6pt}
 \setlength{\belowdisplayskip}{6pt}
\begin{multline}\label{eq:MSD_augm_struct}
    \hat{x}(k+1) = A \hat{x}(k) + B_\mathrm{u} u(k) + \tilde{B}_\mathrm{w}^\mathrm{b} f_{\theta_\mathrm{b}}^\mathrm{FP}(z_\mathrm{b}(k)) +\\ \tilde{B}_\mathrm{w}^\mathrm{a} f_{\theta_\mathrm{a}}^\mathrm{ANN}(\hat{x}_2(k)),
\end{multline}
}\noindent
where the linear parts of the unmodeled terms have been (partly) encoded into the $A$ matrix, and otherwise it was correctly identified that the unmodeled friction effects only depend on the velocity state, i.e.,~$\hat{x}_2(k)$. Due to overlap in representation between the baseline model and the linear parts of $W_{\theta_\mathrm{LFR}}$, $B_\mathrm{u}$ is present in~\eqref{eq:MSD_augm_struct}, however, $\tilde{B}_\mathrm{w}$ approximately satisfies the physical intuition, as $\tilde{B}_\mathrm{w}^\mathrm{b}\approx \left[I_2\, 0_{2\times 1}\right]$ and $\tilde{B}_\mathrm{w}^\mathrm{a}\approx \left[0\,1\right]^\top$. The input of the baseline model also fulfills the expectations, as $z_\mathrm{b}\approx \mathrm{vec}(\hat{x}_1(k), \hat{x}_2(k), u(k))$. 
Additionally, the method detected that the feedback loop is not required, i.e.,~$D_\mathrm{zw}$ was set to zero. Table~\ref{tab:augm_struct_discovery} shows the test accuracy, compared with the baseline model and with the fully parametrized $W_{\theta_\mathrm{LFR}}$ matrix. The applied reweighted regularization scheme not only enhanced physical interpretability but also improved model accuracy.\vspace{-6pt}
%
%
\begin{table}
    \centering
    \scriptsize
    \renewcommand{\arraystretch}{1.1}
    \caption{Test results on the MSD example achieved with and without the automatic model augmentation structure discovery.}
    \label{tab:augm_struct_discovery}
    \begin{tabular}{lc}
    \hline
        Model & Test NRMSE\\
    \hline
         Baseline & 32.11\%\\
         Fully-parametrized WP-LFR & 0.47\%\\
         {\bf Sparse WP-LFR} & {\bf 0.07\%}\\
     \hline
    \end{tabular}\vspace{-1mm}
\end{table}

\subsection{Model order selection}\label{sec:group-lasso-example}\vspace{-6pt}
One could argue that choosing $n_{\mathrm{z_a}}=n_{\mathrm{w_a}}=1$ and $n_\mathrm{x_a}=0$ in the previous example relies on prior physical intuition, which may not always be available. Hence, we demonstrate group-lasso regularization for model order selection on the same identification problem using the same data, baseline model, and learning component parametrization. As the reweighted $\ell_1$ regularization scheme detected that a feedback loop is unnecessary, we set $D_\mathrm{zw}\equiv 0$. Models are trained for 500 Adam epochs and 1000 L-BFGS-B epochs. As a starting point, we set $n_{\mathrm{x_a}}=2$, $n_{\mathrm{w_a}}=3$, and $n_{\mathrm{z_a}}=3$, and follow the group-lasso procedure in \cite{bemporad_efficient_2025} for each variable separately. For each variable, we train with a given regularization weight according to \eqref{eq:group-lasso}, 
then we remove unnecessary dimensions 
via thresholding, and retrain with the reduced dimensions (without regularization). This process is repeated with different $\rho$ values for $x_\mathrm{a}$, $w_\mathrm{a}$, and $z_\mathrm{a}$ until a suitable trade-off between complexity and accuracy is achieved. The results in Fig.~\ref{fig:group-lasso} show that the method, under appropriate regularization weights, recovers the model dimensions from Section~\ref{sec:augm_struct_discovery_example}. In particular, $n_\mathrm{x_a}=2$ leads to overparametrization, while eliminating these redundant variables (setting $n_\mathrm{x_\mathrm{a}}=0$) slightly improves performance. The same can be seen for $n_\mathrm{w_a}$, but interestingly, 
increasing $n_\mathrm{z_a}$ beyond the minimal value does not degrade performance, as presumably, the redundant dimensions improve training of the learning component. Moreover, removing the learning component (setting either $n_\mathrm{w_a}$ or $n_\mathrm{z_a}$ as zero) also yields reasonably accurate models due to the linear components ($A$, $B_\mathrm{u}$, etc.) capturing the viscous effects in $F_\mathrm{fric}$. 
This is an important aspect for applications where less complex control-oriented models are preferred.
\begin{figure*}
    \centering
    \includegraphics{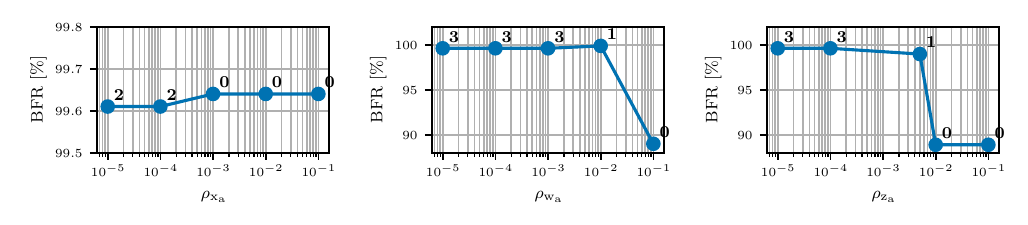}
    \vspace{-18pt}
    \caption{BFR on the test data with different $\rho_\mathrm{x_a}$ (left), $\rho_\mathrm{w_a}$ (center), $\rho_\mathrm{z_a}$ (right) values for group-lasso regularization. The resulting variable dimensions (i.e., $n_\mathrm{x_a}$, $n_\mathrm{w_a}$, and $n_\mathrm{z_a}$) are also shown next to each plot entry.}\vspace{-2mm}
    \label{fig:group-lasso}
\end{figure*}

%% file: sections/7_conclusion.tex
\vspace{-12pt}
\section{Conclusion}\label{sec:conclusion}\vspace{-12pt}
The paper has presented two direct parametrization approaches for a recently proposed general LFR-based model augmentation structure. The first method ensures the well-posedness of the structure while imposing significantly fewer restrictions on the interconnection between the baseline and learning components than existing approaches, thereby increasing the generality of the approach in representing model augmentation structures. The second parametrization guarantees the stability of the resulting models, in addition to well-posedness. Furthermore, we adapted a computationally efficient identification pipeline based on the L-BFGS-B optimizer to enable the use of non-smooth cost function terms for augmentation structure discovery and model order selection. An extensive simulation study and a benchmark identification example both demonstrated the flexibility of the proposed well-posed parametrization and the guaranteed stability of the contracting LFR-based structure. Moreover, we have also highlighted that the adapted identification pipeline can efficiently handle different regularization options, including $\ell_1$ and group-lasso regularization.\vspace{-12pt}

%% file: sections/appendix.tex
\appendix
\section{Proof of Lemma~\ref{lemma:LTI-Lipschitz}}\label{appendix:lemma1}\vspace{-12pt}
First, reformulate the LTI state transition function, as
{\setlength{\abovedisplayskip}{6pt}
 \setlength{\belowdisplayskip}{6pt}
\begin{equation}
    f^\mathrm{FP}_{\theta_\mathrm{b}}(z_\mathrm{b}(k)) = \underbrace{\begin{bmatrix}
        A^\mathrm{FP}_{\theta_\mathrm{b}} & B^\mathrm{FP}_{\theta_\mathrm{b}}
    \end{bmatrix}}_{M_{\theta_\mathrm{b}}} z_\mathrm{b}(k).
\end{equation}
}\noindent
Then, the Lipschitz constant of the linear map $M_{\theta_\mathrm{b}}$: 
{\setlength{\abovedisplayskip}{6pt}
 \setlength{\belowdisplayskip}{6pt}
\begin{equation}
    \sup_{a\neq b} \dfrac{\left\|M_{\theta_\mathrm{b}}a - M_{\theta_\mathrm{b}}b\right\|_2}{\left\|a - b\right\|_2} = \sup_{\xi\neq 0}\dfrac{\left\|M_{\theta_\mathrm{b}}\xi\right\|_2}{\left\|\xi\right\|_2} = \left\|M_{\theta_\mathrm{b}}\right\|_2,
\end{equation}
}\noindent
i.e., $\mathrm{Lip}\{f^\mathrm{FP}_{\theta_\mathrm{b}}\} =\|M_{\theta_\mathrm{b}}\|_2$, and the same can be derived for~$h^\mathrm{FP}_{\theta_\mathrm{b}}$ as well.\placeqed
\vspace{-12pt}